\documentclass[notitlepage,longbibliography,superscriptaddress,pre,nofootinbib]{revtex4-2}
\usepackage{placeins}
\usepackage{booktabs}
\usepackage{pbox}
\usepackage{color}
\usepackage{relsize}
\usepackage{graphicx}
\usepackage[intlimits]{amsmath}
\usepackage{amsxtra,amssymb,fancyhdr, amsthm,latexsym}
\usepackage{cases}
\usepackage{verbatim}
\usepackage{float}
\usepackage[colorlinks=true,linkcolor=blue,citecolor=red]{hyperref}
\usepackage{textcomp}
\usepackage{amsfonts}
\usepackage{amssymb}
\usepackage[mathlines]{lineno}
\usepackage{algpseudocode}

\providecommand{\U}[1]{\protect\rule{.1in}{.1in}}
\newcounter{algcnt}

\newtheorem{theorem}{Theorem}[section]
\newtheorem{lemma}[theorem]{Lemma}

\newtheorem{prop}[theorem]{Proposition}

\renewcommand\thesection{\arabic{section}}
\renewcommand\thesubsection{\thesection.\arabic{subsection}}

\makeatletter
\def\p@subsection{}
\def\p@subsubsection{}
\makeatother

\DeclareMathOperator*{\argmin}{arg\,min}

\providecommand{\func}[1]{\mathop{\mathrm{#1}}\nolimits}

\begin{document}

\begin{empty}

\title{Minimizing capture time with many small traps in
heterogeneous media}

\author{Denis S. Grebenkov}
\affiliation{Laboratoire de Physique de la Mati\`{e}re Condens\'{e}e, CNRS -- Ecole Polytechnique,
Institut Polytechnique de Paris, 91120 Palaiseau, France.}
\email{denis.grebenkov@polytechnique.edu}

\author{Theodore Kolokolnikov}
\affiliation{Dalhousie University Halifax, Nova Scotia, Canada}
\email{tkolokol@gmail.com}

\begin{abstract}
We study the problem of optimally placing a large number of small
absorbing traps to minimize the mean first-passage time (MFPT) of
particles diffusing in a heterogeneous medium with space-dependent
diffusivity $D(x)$ and prescribed initial particle distribution
$\omega(x)$. In two and three dimensions we identify two distinct
regimes, depending on how strongly the trap ensemble depletes the
particles. In the weak trapping regime (fewer traps) the optimal trap
density $\rho$ is the arithmetic average $\rho =
\frac12(\mu+\omega)$ of the normalized inverse diffusivity $\mu
\propto 1/D$ and the initial particle distribution $\omega$. In the strong
trapping regime (more traps), it is proportional to their geometric
average, $\rho \propto \sqrt{\mu \omega}$. We revisit the classical
Green's function approach and show that it applies to weak trapping
only; in 2D this requires the trap size to be exponentially small in
the number of traps, so that most applications of interest fall in the
strong trapping regime instead. We therefore develop a homogenization
approach that captures both regimes.  A special choice of initial
particle distribution is $\omega \propto 1/D$, which corresponds to
the equilibrium particle distribution in the absence of traps
according to the It\^{o} interpretation. In this case we show that the
optimal trap distribution is the stationary particle distribution
itself, regardless of the trapping strength, and this makes the MFPT
constant throughout the domain; the same trap distribution also
minimizes the worst-case capture time.  Finally, in one dimension we
find instead $\rho\propto\left(\mu \omega\right) ^{1/3}$, and a
similar analysis covers space-dependent drift and thin domains of
variable cross-section.  Direct numerical optimization confirms the
analytical results.
\end{abstract}

\maketitle
\end{empty}

\section{Introduction}

The mean first-passage time (MFPT) of a diffusing particle to a set of small
absorbing traps is a classical problem in diffusion theory (also called the
narrow escape problem \cite{holcman2014narrow}), with applications ranging
from cellular signaling and receptor-ligand binding \cite%
{berg1977physics,lauffenburger1993receptors,Holcman2013,Bressloff2013} to
search theory and ecology \cite%
{viswanathan1999optimizing,viswanathan2008levy,benichou2014firstpassage,GrebenkovBook}%
. A recurring question is: given a fixed number $N$ of traps, how should
they be arranged in space to minimize the average time it takes a diffusing
particle to be captured? For homogeneous media, this question has been
studied extensively in various contexts \cite%
{kuijlaars1998asymptotics,kolokolnikov2005optimizing,cheviakov2013narrow,lindsay2017first,GrebWard2026,GrebWard2026b}%
.

However, most of the natural environments (e.g., the cytoplasm of a living
cell) are not homogeneous. Several recent works have looked at the MFPT
problem with inhomogeneities \cite%
{Hwang2012,vaccario2015first,Godec2015,Godec2016,Grebenkov2016,bressloff2017temporal,Lanoiselee2018,vuijk2018pseudochemotaxis,Grebenkov2018,Mangeat2019,Mangeat2021,caprini2022dynamics,Chun2023,tunglawley2026stochastic,tunglawley2025escape}%
. Most of these works consider a single or several traps whose locations are
fixed. However, the fundamental problem of finding the optimal spatial
configuration of traps minimizing the MFPT remains open. When the medium is
homogeneous, the optimal configuration of traps is expected to be roughly
uniform \cite{cheviakov2013narrow,varizy2026pattern}. In turn, in the
presence of heterogeneous diffusion or drift (e.g., advection, external
force), the optimal trap location is intuitively expected to have a higher
density in areas where particles ``linger'', i.e., in areas of lower
diffusion or slower advection. The goal of this paper is to develop precise
asymptotic descriptions of optimal configurations for several model
situations in the limit of many traps.

Consider a bounded domain $\Omega\subset{\mathbb{R}}^{d}$ with a smooth
boundary $\partial\Omega$. For a smoothly varying space-dependent diffusion
coefficient $D(x)>0$, the Fokker--Planck equation in the It\^{o}
interpretation reads 
\begin{equation}  \label{eq:Fokker}
p_{t}=\Delta\left( D(x)p\right) ,
\end{equation}
where $\Delta$ is the Laplace operator \cite{Gardiner,Risken,Redner}. In
turn, the associated MFPT is governed by the adjoint operator (backward
form). If the diffusing particle starts from a point $x\in\Omega$, the MFPT $%
u(x)$ to any of $N$ small non-overlapping spherical traps, all of the same
radius $\varepsilon,$ located at points $x_{j}\in\Omega$ ($j=1\ldots N$),
satisfies 
\begin{equation}
\left\{ 
\begin{array}{rll}
\Delta u & =-M(x) & x\in\Omega\backslash\left( \cup_{j=1}^{N}B_{\varepsilon
}(x_{j})\right) , \\ 
u & =0 & x\in\left( \cup_{j=1}^{N}\partial B_{\varepsilon}(x_{j})\right) ,
\\ 
\partial_{n}u & =0, & x\in\partial\Omega,%
\end{array}
\right.  \label{u2d}
\end{equation}
where $\partial_{n}$ is the normal derivative, $M(x)=1/D(x)$ (inverse
diffusivity or mobility), and $B_{\varepsilon}(x_{j})$ is a ball of radius $%
\varepsilon$ centered at $x_{j}:B_{\varepsilon}(x_{j})=\left\{ x\in \mathbb{R%
}^d:~~\left\vert x-x_{j}\right\vert <\varepsilon\right\} $. We also define 
\begin{align}
\left\vert \Omega\right\vert & =\int_{\Omega}dx,\qquad \left\vert
M\right\vert =\int_{\Omega}M(x)dx = \int_{\Omega} \frac{dx}{D(x)} \,, \qquad
\mu(x) =\frac{M(x)}{\left\vert M\right\vert }=\frac{D^{-1}(x)}{\left\vert
D^{-1}\right\vert }.  \label{alpha}
\end{align}
Note that $\mu(x)$ is precisely the equilibrium distribution of a particle
when there is no trap. Indeed, the probability flux density, $\mathbf{J}_{%
\mathrm{eq}} = - \nabla(Dp_{\mathrm{eq}})$, is zero at equilibrium, implying 
$Dp_{\mathrm{eq}} =\func{const}$. Normalizing so that $\int\nolimits_{%
\Omega} p_{\mathrm{eq}} =1$ gives exactly $p_{\mathrm{eq}}(x) =\mu(x)$.

We focus on the average capture time. Let $\omega(x)$ denote a given
distribution of the particle's starting location, normalized throughout to
be a probability density, 
\begin{equation}
\int_{\Omega}\omega(x)\,dx=1,\qquad\omega\geq0,  \label{omeganorm}
\end{equation}
so that the quantity to be minimized is the $\omega$-averaged MFPT 
\begin{equation}
\overline{u}_{\omega}=\int_{\Omega}\omega(x)u(x)\,dx.  \label{ubarw}
\end{equation}
The uniform choice $\omega = \omega_0 \equiv1/\left\vert \Omega\right\vert $
gives the \textquotedblleft classical\textquotedblright\ volume-averaged
MFPT $\overline{u}=\int_{\Omega}u\,dx/\left\vert \Omega\right\vert $,
averaged over all starting locations distributed uniformly. Our main
question is the following:

\textbf{Question:}\emph{\ In the limit of large }$N$\emph{\ and small }$%
\varepsilon,$\emph{\ for a given diffusivity} $D(x)$ \emph{and a given
starting distribution }$\omega(x)$\emph{, how should the traps be
distributed to minimize the }$\omega$\emph{-averaged MFPT\ }$\overline{u}%
_{\omega}$?

We consider this problem in one, two or three dimensions. Let us summarize
the results in 2D and 3D first. The answer strongly depends on the relative
scaling of $N$ and $\varepsilon$. Define 
\begin{equation}
\ \nu:=\left\{ 
\begin{array}{c}
4\pi\varepsilon\text{ in 3D} \\ 
\frac{2\pi}{\log(1/\varepsilon)} \text{ in 2D}%
\end{array}
\right. ,\qquad \kappa:=N\nu.  \label{nu}
\end{equation}
We refer to $\kappa=N\nu$ as the trapping strength. In particular, we
identify two opposite regimes, $\kappa\ll1$ or $\kappa\gg1,$ which we refer
to as weak or strong trapping, respectively.

Our approach starts with a well-known Green's function energy-type
formulation for the MFPT in the limit of small traps \cite%
{kolokolnikov2005optimizing,cheviakov2024optimization}. Taking a continuum
limit as $N\rightarrow\infty$ we then obtain a simple asymptotic description
of how the optimal trap density depends on the heterogeneity $\mu(x)$ and on
the starting distribution $\omega(x)$ (see \S \,\ref{sec:2D}): 
\begin{equation}
\rho(x)\sim\frac{1}{2}\left( \mu(x)+\omega(x)\right) ,\quad\kappa \ll1,\ \
N\gg1.  \label{rhoavg}
\end{equation}
Here and below, we use $f\sim g$ to mean $f/g\rightarrow1$ in the indicated
limit. For a uniformly distributed starting point, $\omega_0
\equiv1/\left\vert \Omega\right\vert $, this is the average of $\mu$ and the
uniform density. The corresponding optimal $\overline{u}_{\omega}$ is given
by (\ref{ubarsmallk}) below.

The optimal density $\rho(x)$ (in distribution sense)\ is defined as
follows. For a fixed $N,$ let $x_{1}\ldots x_{N}$ be optimal trap locations
that minimize the MFPT\ for the problem (\ref{u2d}) in the limit $%
\varepsilon\rightarrow0$, and $\rho_{N}(x)=\frac{1}{N}\sum
_{j=1}^{N}\delta\left( x-x_{j}\right)$. Note $\int_{{S}}\rho_{N}(x)\,dx$
gives the fraction of $N$ traps inside an arbitrary region ${S}%
\subset\Omega. $ We then define the limiting trap density as $%
\rho=\lim_{N\rightarrow\infty }\rho_{N},$ understood in the weak sense, so
that $\int_{{S}}\phi(x)\rho _{N}(x)\,dx\rightarrow\int_{{S}}\phi(x)\rho(x)dx$
as $N\rightarrow\infty$, for any appropriate test function $\phi(x)$ {and
subset $S\subset \Omega$.}

The downside of the Green's function approach is that it is only valid in
the weak trapping regime $\kappa\ll1$. This is particularly restrictive in
two dimensions as it requires trap sizes $\varepsilon$ to be exponentially
small in $N$, which is unrealistic for a physically or numerically
reasonable trap size (e.g. even for $N = 10$, one needs $\varepsilon \ll
e^{-2\pi N} \approx 5\cdot 10^{-28}$). Most realistic 2D situations are
therefore in the strong trapping regime $\kappa\gg1.$ In 3D, this
restriction is milder, but still relevant under many realistic scenarios
(e.g. $\varepsilon=0.01,N=100$ gives $\kappa\approx12.6$).

To overcome this restriction, we make use of a homogenization technique \cite%
{cioranescu1982strange,marchenko2006homogenization,berezhkovskii2004boundary,muratov2008boundary,cheviakov2013narrow}%
. It originates in the study of domains perforated by many small holes,
developed by Marchenko and Khruslov \cite%
{marchenko1964boundary,marchenko2006homogenization}, by Kac \cite%
{kac1974probabilistic} and by Rauch and Taylor \cite{rauch1975potential} for
the \textquotedblleft crushed ice\textquotedblright\ problem, by
Papanicolaou and Varadhan \cite{papanicolaou1980diffusion}, and by
Cioranescu and Murat \cite{cioranescu1982strange}, who coined
\textquotedblleft strange term\textquotedblright\ (terme \'{e}trange)
corresponding to the non-vanishing $O(1)$ effect of the holes in the limit $%
\kappa =O(1)$. Berezhkovskii, Muratov, Shvartsman and co-workers \cite%
{berezhkovskii2004boundary,berezhkovskii2006homogenization,muratov2008boundary}
applied the same idea to surfaces bearing many small traps, and Cheviakov
and Zawada \cite{cheviakov2013narrow} to the narrow escape problem. Here the
novelty is that the trap density is neither uniform, nor prescribed, but is
the quantity to be optimized. The homogenization limit for this problem,
derived in \S \thinspace \ref{sec:homog}, is 
\begin{equation}
\ \ \Delta u-\kappa \rho (x)u(x)=-M(x)\text{ in }\Omega ,\qquad \partial
_{n}u=0\text{ on }\partial \Omega .  \label{homog}
\end{equation}

In contrast to the Green's function approach, the homogenized formulation is
not restricted to $\kappa\ll1$ and provides an effective description across
the weak-, intermediate-, and strong-trapping regimes, subject to the
dilute-trap assumptions $\varepsilon^{d}N\ll1$, underlying the
homogenization limit. It reproduces the Green's function result (\ref{rhoavg}%
) in the weak-trapping case $\kappa\ll1,$ and we further obtain the
corresponding $\omega$-averaged MFPT: 
\begin{equation}
\overline{u}_{\omega}=\frac{\left\vert M\right\vert }{\kappa}{-}\frac{%
\left\vert M\right\vert }{4}\int_{\Omega}\left\vert \nabla
V_{\omega-\mu}\right\vert ^{2}\, +O(\kappa) ,  \label{ubarsmallk}
\end{equation}
where $V_{\omega-\mu}$ is the Poisson potential defined in (\ref{Vweq}).

Moreover, homogenization also gives an explicit answer for $\kappa\gg1.$ In
the latter case, we find that the optimal density scales like a \emph{%
geometric }average of $\omega$ and $\mu$: 
\begin{equation}
\rho(x)\sim\frac{\sqrt{\omega(x){\mu}(x)}}{\int_{\Omega}\sqrt {\omega(x){\mu}%
(x)}dx},\qquad\kappa\gg1 .  \label{rhogeom}
\end{equation}
The corresponding optimal $\overline{u}_{\omega}$ is given by%
\begin{equation}
\overline{u}_{\omega}\sim\frac{\left\vert M\right\vert }{\kappa}\left(
\int_{\Omega}\sqrt{\omega(x){\mu}(x)}dx\right) ^{2},\qquad\kappa\gg1.
\label{ubarwklarge}
\end{equation}
More generally, we derive an optimality system (\ref{kappa-any}) valid for
any $\kappa.$ We use this system to validate our results numerically
throughout.

A very natural choice is $\omega =\mu ,$ corresponding to the equilibrium
particle distributions in the absence of traps; in this case we simply get $%
\rho =\mu $, both from (\ref{rhoavg}) or (\ref{rhogeom}), and $\overline{u}%
_{\mu }=\frac{\left\vert M\right\vert }{\kappa }$ in both regimes. In fact
we show that this holds for any $\kappa .$ This is one of the key findings
of the paper which we summarize as follows.

\begin{theorem}
\label{thm:alpha} Consider a Brownian particle diffusing without traps under
a heterogeneous diffusivity $D(x)$ in 2D or 3D, whose density is governed by
the Fokker--Planck equation (\ref{eq:Fokker}) (in It\^{o} sense). Wait long
enough until its distribution equilibrates to $\mu(x)$ given by (\ref{alpha}%
). Now turn on the traps. Then the trap density minimizing the mean capture
time starting from $\mu(x)$ is precisely $\rho(x)=\mu(x)$: the traps should
be distributed exactly like the particle itself. This holds for every $%
\kappa=\nu N$. At the optimum, the MFPT $u(x)$ is constant, independent of
the starting location:%
\begin{equation}
u_{\mathrm{opt}}(x)\equiv\frac{\left\vert M\right\vert }{\kappa}\,,
\label{mainresult}
\end{equation}
which is therefore also the mean capture time. Moreover, the optimal density 
$\rho=\mu$ is the unique global minimum of $\overline{u}_\mu$:%
\begin{equation}
\ \int_{\Omega}\mu u\geq\frac{\left\vert M\right\vert }{\kappa}\ \text{for
all densities }\rho {\text{ and any}\ \kappa,}  \label{eq:thmalpha}
\end{equation}
with equality if and only if $\rho=\mu.$
\end{theorem}

Instead of minimizing the \emph{average} $\overline{u}_{\omega}$ for a given 
$\omega$, one can ask how to minimize the \emph{worst-case capture time} $%
\max\limits_{x\in\Omega}u(x).$ In fact Theorem \ref{thm:alpha} gives the
answer: since $\int\mu u\leq\max u$ (average with respect to any measure is
bounded by $\max u$), it then immediately follows from Theorem \ref%
{thm:alpha} that the choice $\rho=\mu$ solves the mini-max problem to
minimize the worst-case (i.e. maximum) MFPT. We state this as a theorem.

\begin{theorem}
\label{thm:minimax} The worst-case mean capture time $\max\limits_{x\in
\Omega }u(x)$ is minimized by choosing $\rho =\mu ,$ in which case $u(x)$ is
everywhere constant: $u(x)=\frac{\left\vert M\right\vert }{\kappa }$. In
other words, $\rho =\mu $ solves the mini-max problem $\min\limits_{\rho 
\text{,\ density}}\left\Vert u\right\Vert _{\infty }.$
\end{theorem}

Finally, we have the following generalization of Theorem \ref{thm:alpha} for
any starting density $\omega:$

\begin{theorem}
\label{thm:general} Define $\overline{u}_{\omega}[\rho]=\int_{\Omega}%
\omega(x)u(x)dx$, where $u$ solves (\ref{homog}) for a given density $\rho.$
Then 
\begin{equation}
\overline{u}_{\omega}[\rho]\geq\frac{\left\vert M\right\vert }{\kappa}\left(
\int_{\Omega}\sqrt{\omega\mu}\right) ^{2}\text{ for all densities }%
\omega,\rho {\text{ and any}\ \kappa.}  \label{eq:thmgeneral}
\end{equation}
Moreover the equality is achieved if and only if $\omega=\rho=\mu$. 
\end{theorem}

Note that the lower bound is the same as in (\ref{ubarwklarge}):\ it is
approached as $\kappa\rightarrow\infty$ but never touched unless $\omega
=\rho=\mu.$ Indeed, Theorem \ref{thm:alpha} can be viewed as the special
case $\omega=\mu$ of Theorem \ref{thm:general}.

The proofs of Theorems \ref{thm:alpha} and \ref{thm:general} are given in \S %
\,\ref{sec:proofs}.

In \S \,\ref{sec:1D} we consider the analogous problem in one dimension.
Here, the problem simplifies, since one-dimensional traps split the domain
into disjoint subintervals. We find that the optimal trap density scales
with $\mu^{1/3}$ for the uniform starting distribution. On the other hand if
the particle starts from the equilibrium distribution $\omega=\mu$, the
optimal distribution becomes proportional to $\mu^{2/3}.$ Both are special
cases of a single formula: for an arbitrary starting distribution $%
\omega(x), $ the optimal density in 1D is proportional to $\left(
\mu(x)\omega(x)\right) ^{1/3},$ recovering $\mu^{1/3}$ for uniform $\omega$
and $\mu^{2/3}$ for $\omega=\mu.$ This is in contrast to 2D, where the same
choice $\omega=\mu$ makes the optimal density exactly proportional to $\mu$
itself, for every value of $\kappa = \nu N$. Finally, we also apply a
similar machinery to determine optimal configurations for heterogeneous
media with spatially-dependent drift, as well as to ``thin'' domains. 

\section{Green's function approach}

\label{sec:2D}

We start by generalizing the classical Green's function approach to MFPT\
computations with heterogeneous diffusion and starting positions with the
goal of optimizing (\ref{ubarw}) over trap positions in 2D or 3D. 

To show our main result (\ref{rhoavg}), we first use standard asymptotics to
reduce the MFPT\ to an energy formulation involving Green's functions \cite%
{kolokolnikov2005optimizing}, which we generalize here to the case of
non-constant $M(x)$ and $\omega(x)$. The argument proceeds in two steps: we
first derive (following \cite{kolokolnikov2005optimizing}) a discrete
pairwise Green's-function energy $E(x_{1},\ldots,x_{N})$ whose minimizer
approximates the optimal trap locations (Lemma \ref{lemma:E2d} below), we
then pass to the continuum limit $N\rightarrow\infty$ of this energy to
obtain the density formula (\ref{rhoavg}).

We first introduce some notation. Let $G(x,y)$ denote the Neumann Green's
function on the domain $\Omega$ satisfying%
\begin{equation}
\left\{ 
\begin{array}{c}
\Delta G=\frac{1}{\left\vert \Omega\right\vert }-\delta(x-y),\quad
x,y\in\Omega\subset\mathbb{R}^{{d}}; \\ 
\partial_{n}G=0\text{ \ on }\partial\Omega;\ \ \ \ \int_{\Omega}G(x,y)dx=0,%
\end{array}
\right.  \label{G}
\end{equation}
and let $G_{\mu}$ be the analogous Green's function with a heterogeneity,
defined via%
\begin{equation}
\left\{ 
\begin{array}{c}
\Delta G_{\mu} =\mu(x)-\delta(x-y),\quad x,y\in\Omega\subset \mathbb{R}^{{d}%
}; \\ 
\partial_{n}{G_{\mu}}=0\text{ \ on }\partial\Omega;\ \ \ \int_{\Omega }{%
G_{\mu}}(x,y)dx=0.%
\end{array}
\right.  \label{GA}
\end{equation}
Finally, let $R$ and $R_{\mu}$ be the regular parts of $G$ and $G_{\mu} $,
respectively: 
\begin{align}
R_{\mu}(x,y) & = G_{\mu}(x,y) - g(x-y),  \label{RA} \\
R(x,y) & =G(x,y)- g(x-y),  \label{R}
\end{align}
where $g$ is the free-space Green's function in $\mathbb{R}^d$: 
\begin{equation}
g(x)=\left\{ 
\begin{array}{ll}
\displaystyle -\frac{1}{2\pi}\log\left\vert x\right\vert , & d=2, \\ 
\displaystyle \frac{1}{4\pi\left\vert x\right\vert }, & d=3.%
\end{array}
\right.  \label{Sd}
\end{equation}
It is convenient to introduce, for an arbitrary function $\psi$ on $\Omega$,
the potential generated by $\psi$, 
\begin{equation}
V_{\psi}(x):=-\int_{\Omega}\psi(y)G(x,y)\,dy,  \label{Vw}
\end{equation}
which, by (\ref{G}), is the unique mean-zero solution of 
\begin{equation}
\Delta V_{\psi}=\psi(x)-\frac{\left\vert \psi\right\vert }{\left\vert
\Omega\right\vert },\qquad\partial_{n}V_{\psi}=0\text{ on }\partial
\Omega,\qquad\int_{\Omega}V_{\psi}(x)dx=0.  \label{Vweq}
\end{equation}
The two densities entering our problem generate the potentials $V_{\mu}$ and 
$V_{\omega}$ respectively: the first comes from heterogeneity, the second
from the starting distribution. Note also that $V_{\psi}$ is linear in $\psi$
so that e.g. $V_{\mu}+V_{\omega}=V_{\mu+\omega}.$ For a uniform starting
distribution $\omega\equiv1/\left\vert \Omega\right\vert $ one has $%
V_{\omega }\equiv0$.

Note that $G_{\mu}$ can be decomposed as follows:%
\begin{equation}
G_{\mu}(x,y)=G(x,y)+V_{\mu}(x),  \label{Ga=G+U}
\end{equation}
where $G$ satisfies (\ref{G}) and $V_{\mu}$ is given by (\ref{Vw}) with $%
\psi=\mu$.

We start with the following asymptotic reduction for $u(x)$. In the outer
region, away from the traps $x_{j}$, we search for $u$ in the form 
\begin{equation}
u(x)\sim T+\sum_{j=1}^{N}C_{j}G_{\mu}(x,x_{j}),  \label{145}
\end{equation}
where the constants $T$ and $C_{j}$ are to be found. Note that this estimate
is asymptotically valid away from $x_{j}$ and corresponds to a monopole
approximation of the trap; the neglected terms are $O(\varepsilon^{2})$
coming from neglecting the dipole terms, and independent of $\nu$ to all
orders in 2D \cite%
{ward1993strong,ward1993summing,lindsay2015narrow,lindsay2025orientation}.

For $x\notin\left\{ x_{1},\ldots,x_{N}\right\} $, Eqs. (\ref{u2d}, \ref{GA})
imply $\Delta u=-M(x)=\sum_{j=1}^{N}C_{j}\mu(x).$ Integrating over $\Omega $
we then obtain 
\begin{equation}
\sum\limits_{j=1}^{N}C_{j}=-\left\vert M\right\vert .  \label{139}
\end{equation}

In the inner region near $x_{k}$, we change variables $x=x_{k}+\varepsilon
y, $ $u(x)=\tilde{u}_k(y)$ to obtain, at leading order, $\Delta_{y}\tilde{u}%
_k\sim0, $ which is to be solved outside the unit ball, i.e., $\left\vert
y\right\vert >1$, with $\tilde{u}_k =0$ on $\left\vert y\right\vert =1$ and $%
\tilde{u}_k$ growing at most like $g$ at infinity. The admissible solutions
are 
\begin{equation}  \label{eq:uk_inner}
\tilde{u}_k(y)\sim A_k \times \left\{ 
\begin{array}{ll}
\log\left\vert y\right\vert , & d=2, \\ 
\bigl(1-\left\vert y\right\vert ^{-1}\bigr) , & d=3,%
\end{array}
\right.
\end{equation}
with the factor $A_k$ to be found by matching the inner and outer solution.
Writing the outer problem with respect to the inner variables near $x_{k}$,
we have:%
\begin{align}
u(x_{k}+\varepsilon y) & \sim T+\sum_{j=1}^{N}C_{j}
G_{\mu}(x_{k}+\varepsilon y,x_{j})  \notag \\
& \sim T+\sum_{j\neq k}^{N}C_{j} G_{\mu}(x_{k},x_{j})+C_{k} \bigl[ %
g(\varepsilon y) + R_{\mu}(x_{k},x_{k})\bigr] .  \label{1234y}
\end{align}
{Using (\ref{Sd}) and matching 
(\ref{eq:uk_inner}) and (\ref{1234y}), one fixes $A_{k}=-C_{k}/(2\pi)$ in 2D
and $A_{k}=-C_{k}/(4\pi\varepsilon)$} in 3D; evaluating at the trap $%
\left\vert y\right\vert =1$ in both cases yields 
\begin{equation}
0 {=} T+\sum_{j\neq k}^{N}C_{j}G_{\mu}(x_{k},x_{j})+C_{k}\left[ 
\frac {1}{\nu}+R_{\mu}(x_{k},x_{k})\right] ,\ \ \ k=1\ldots N,  \label{140}
\end{equation}
with $\nu$ defined in (\ref{nu}). 

There are two ways to proceed. In the first approach, one can consider (\ref%
{140}) together with (\ref{139}) as a linear system of $N+1$ equations for
the $N+1$ unknowns $T$ and $C_{1},\ldots ,C_{N}$, which can be easily solved
numerically; this approach is valid for any trapping strength $\kappa $.
From (\ref{145}) and using $\int_{\Omega }\omega G_{\mu
}(x,x_{j})dx=-V_{\omega }(x_{j})+\int_{\Omega }\omega V_{\mu }$ together
with (\ref{139}), we then obtain a simplified expression 
\begin{equation}
\overline{u}_{\omega }=T-\sum_{j=1}^{N}C_{j}V_{\omega }(x_{j})-\left\vert
M\right\vert \int_{\Omega }\omega V_{\mu }.  \label{fullubar}
\end{equation}%
This expression is the essence of a \emph{hybrid asymptotic-numerical method}
for computing $\overline{u}_{\omega }$: \emph{numeric} because the
coefficients $C_{j}$ and $T$ must be computed numerically by solving (\ref%
{139}, \ref{140}); but \emph{asymptotic} in that the traps in the full
problem (\ref{u2d})\ are replaced by sources and the entire solution is
approximated by the Neumann Green's functions. This method has been
extensively used by Ward and collaborators \cite%
{ward1993strong,ward1993summing,ward2018spots,kolokolnikov2005optimizing,pillay2010asymptotic,cheviakov2010asymptotic}%
, who coined the term \textquotedblleft hybrid\textquotedblright ; also
referring to it as the \textquotedblleft sum-of-logs\textquotedblright\ when
done in 2D. Correspondingly, we shall refer to (\ref{fullubar}), together
with (\ref{140}, \ref{139}), as the \textbf{hybrid system.}

The accuracy of approximation (\ref{fullubar})\  inherits the accuracy of
the approximation (\ref{145}), which is of $O(\varepsilon ^{2}),$ coming
from dropping the dipole terms. As a consequence, the error between (\ref%
{fullubar})\ and (\ref{u2d}) is of $O(\varepsilon ^{2})$ even in the strong
trapping regime $\nu N\gg 1.$

In turn, the second approach corresponds to the asymptotic expansion in $%
\nu\ll1$ of the hybrid system (\ref{fullubar}); as we will see it is only
valid in the weak trapping regime $\nu N\ll1.$ We now expand in $\nu$, 
\begin{equation}
C_{k}=C_{k0}+C_{k1}\nu+C_{k2}\nu^{2} + \ldots, \qquad T=\frac{1}{\nu}%
T_{0}+T_{1}+\nu T_{2} + \ldots  \label{T}
\end{equation}
We then obtain from (\ref{140}), at leading order, 
\begin{equation*}
C_{k0}=-T_{0} \qquad k = 1,\ldots,N.
\end{equation*}
Upon substituting into (\ref{139})\ this yields%
\begin{equation}
T_{0}=\frac{1}{N}\left\vert M\right\vert ,\ \ \ C_{k0}=-\frac{1}{N}%
\left\vert M\right\vert , \qquad k = 1,\ldots,N.  \label{142}
\end{equation}
At the next order in $\nu,$ (\ref{140}) yields%
\begin{equation}
T_{1}+C_{k1}+\sum_{j\neq k}^{N}C_{j0}G_{\mu}(x_{k},x_{j})+C_{k0}R_{\mu
}(x_{k},x_{k})=0,\qquad k=1\ldots N,  \label{141}
\end{equation}
and from (\ref{139}) we get 
\begin{equation}
\sum_{k=1}^{N}C_{k1}=0.  \label{eq:Ck1_sum}
\end{equation}
Summing $N$ equations (\ref{141}) for $k=1\ldots N$ and using (\ref{142})
and (\ref{eq:Ck1_sum}), we then obtain 
\begin{equation*}
T_{1}=\frac{|M|}{N^{2}}\sum_{j=1}^{N}\sum_{k=1}^{N}G_{\mu kj},\qquad\text{%
where }G_{\mu kj}=\left\{ 
\begin{array}{c}
G_{\mu}(x_{k},x_{j}),\ \ k\neq j \\ 
R_{\mu}(x_{k},x_{k}),\ \ k=j%
\end{array}
\right. .
\end{equation*}
The decomposition (\ref{Ga=G+U}) also yields 
\begin{equation*}
T_{1}=\frac{\left\vert M\right\vert }{N^{2}}\left\{
\sum_{k=1}^{N}\sum_{j\neq
k}^{N}G(x_{k},x_{j})+\sum_{k=1}^{N}R(x_{k},x_{k})+N\sum_{k=1}^{N}V_{\mu
}\left( x_{k}\right) \right\} .
\end{equation*}
An analogous expansion shows that 
\begin{equation}
T_{2}=\frac{\left\vert M\right\vert }{N^{3}}\left( \sum_{k=1}^N \sum_{j=1}^N
G_{\mu kj}\right)^{2} - \frac{\left\vert M\right\vert}{N^{2}}
\sum_{k=1}^N\sum_{j=1}^N \sum_{l=1}^N G_{\mu jl} G_{\mu kj}\,,  \label{T2}
\end{equation}
so that, generically, $T_{2}=O(N).$ It follows that the entire asymptotic
procedure is only valid for $\nu N\ll1$, so as not to contradict the
asymptotic orders in (\ref{T}).

It remains to compute $\overline{u}_{\omega}$. From (\ref{145}), we have 
\begin{equation*}
\overline{u}_{\omega} = \int_{\Omega}\omega
u=T+\sum_{j=1}^{N}C_{j}\int_{\Omega}\omega(x){G_{\mu}}(x,x_{j})dx,
\end{equation*}
and by (\ref{Ga=G+U}), (\ref{Vw}) and the symmetry of $G$, 
\begin{equation*}
\int_{\Omega}\omega(x){G_{\mu}}(x,x_{j})dx=-V_{\omega}(x_{j})+\int_{\Omega
}\omega{V_{\mu}}.
\end{equation*}
Using 
(\ref{139}) and 
(\ref{142}), we obtain 
\begin{equation}
\overline{u}_{\omega} \sim T+\frac{\left\vert M\right\vert }{N}%
\sum_{k=1}^{N}V_{\omega}(x_{k})-\left\vert M\right\vert \int_{\Omega}\omega
V_{\mu}.  \label{ubarwexp}
\end{equation}
Note that the last two terms in (\ref{ubarwexp})\ are of the same order as
the terms in $T_{1}.$ Together with $T\sim T_{0}/\nu+T_{1}$ this gives the
following.

\begin{lemma}
\label{lemma:E2d} Define a pairwise Green's function energy to be%
\begin{equation}
E(x_{1},\ldots x_{N})=\sum_{k=1}^{N}\sum_{j\neq
k}^{N}G(x_{k},x_{j})+\sum_{k=1}^{N}R(x_{k},x_{k})+N\sum_{k=1}^{N}{%
V_{\mu+\omega}(x_{k})-{N^{2}}\int_{\Omega}\omega{V_{\mu}}},  \label{greenE}
\end{equation}
where $G$ and $R$ are defined in (\ref{G}, \ref{R}), $V_{\mu+\omega}$ in (%
\ref{Vw}) with $\psi=\mu+\omega$, and $\nu$ in (\ref{nu}). In the limit $%
\varepsilon\rightarrow0$ and for a fixed $N,$ we have%
\begin{equation}
{\overline{u}_{\omega}=\int_{\Omega}\omega(x)u(x)dx\sim\frac{\left\vert
M\right\vert }{\nu N}+\frac{\left\vert M\right\vert }{N^{2}}E(x_{1},\ldots
x_{N})+{\mathcal{O}}\left( \nu N\right) }\text{ \ \ as }\varepsilon
\rightarrow0.
\end{equation}
\end{lemma}

The optimal trap locations correspond to the minimum of the pairwise Green's
function energy $E$ in (\ref{greenE}). Note that we must have $\nu N\ll1$
for the asymptotics to be applicable. This is quite restrictive especially
in 2D: $\varepsilon$ must be exponentially small in $N$, see (\ref{nu}).

To compute optimal configurations of the energy $E$, we will study its
gradient flow, as was done e.g. in \cite%
{Serfaty2020,kolokolnikov2014tale,varizy2026pattern,fetecau2011swarm,kolokolnikov2011stability}%
. We introduce a time-dependent position $x_{k}(t)$ (with artificial time $t$%
) along with the gradient flow $x_{k}^{\prime}(t)=-\partial_{x_{k}}
E(x_{1}(t),\ldots x_{N}(t)).$ Then $x_{k}(t)\rightarrow x_{k}$ as $%
t\rightarrow\infty$, where $x_{k}$ is a local minimum of $E.$ We now compute%
\begin{equation*}
\partial_{x_{k}}E = 2\left( \sum_{j\neq
k}^{N}\nabla_{x}G(x_{k},x_{j})\right) +2\nabla_{x}R(x_{k},x_{k})+N\nabla
V_{\mu+\omega }(x_{k}) .
\end{equation*}
The associated gradient flow (up to time rescaling by $\frac{1}{2N}$ for
convenience) then becomes%
\begin{equation}
-x_{k}^{\prime}(t)=\frac{1}{N}\left( \sum_{j\neq
k}\nabla_{x}G(x_{k},x_{j})+\nabla_{x}R(x_{k},x_{k})\right) +\frac{1}{2}%
\nabla V_{\mu+\omega }(x_{k}).  \label{ode}
\end{equation}
The {global} minimizer of $E$ corresponds to a stable equilibrium of the ODE
system (\ref{ode}):\ $x_{k}=\lim_{t\rightarrow\infty}x_{k}(t)$. We now look
at the continuum density formulation in the limit $N\rightarrow\infty$ of
the steady state of (\ref{ode}). Following \cite%
{kolokolnikov2014tale,varizy2026pattern}, we define a coarse-grained density
to be%
\begin{equation*}
\rho(x,t)=\frac{1}{N}\sum\limits_{j=1}^{N}\delta(x-x_{j}(t)).
\end{equation*}
Then (\ref{ode}) is equivalent to the continuum density evolution equation $%
\rho_{t}+\nabla \cdot \left( v\rho\right) =0$, where $v$ is the associated velocity
field:%
\begin{equation}
v(x)=\frac{1}{2}\nabla{V_{\mu+\omega}(x)}+\int_{\Omega}\nabla_{x}G(x,y)%
\rho(y)dy.  \label{153}
\end{equation}
At the steady state, we have $v(x)=0$ so that $\nabla\cdot v=0.$ Applying
the divergence to both sides of (\ref{153}) and using (\ref{Vweq}) and (\ref%
{G}), we obtain%
\begin{equation*}
0=\nabla\cdot v=\frac{1}{2}\Delta V_{\mu+\omega}(x) + \int_{\Omega}\Delta
G(x,y)\rho(y)dy=\frac{1}{2} \left( \mu(x)+\omega(x)\right) -\rho(x) .
\end{equation*}
It follows that $\rho(x)=\frac{1}{2}\left( \mu(x)+\omega(x)\right) $, which
is (\ref{rhoavg}). $\blacksquare$

\subsection{Numerics}

We now consider several numerical experiments for a disk and a rectangular
domain, for which an efficient computation of the Neumann Green's function $%
G $ is available.

\begin{figure}[tb]
\includegraphics[width=0.95\textwidth]{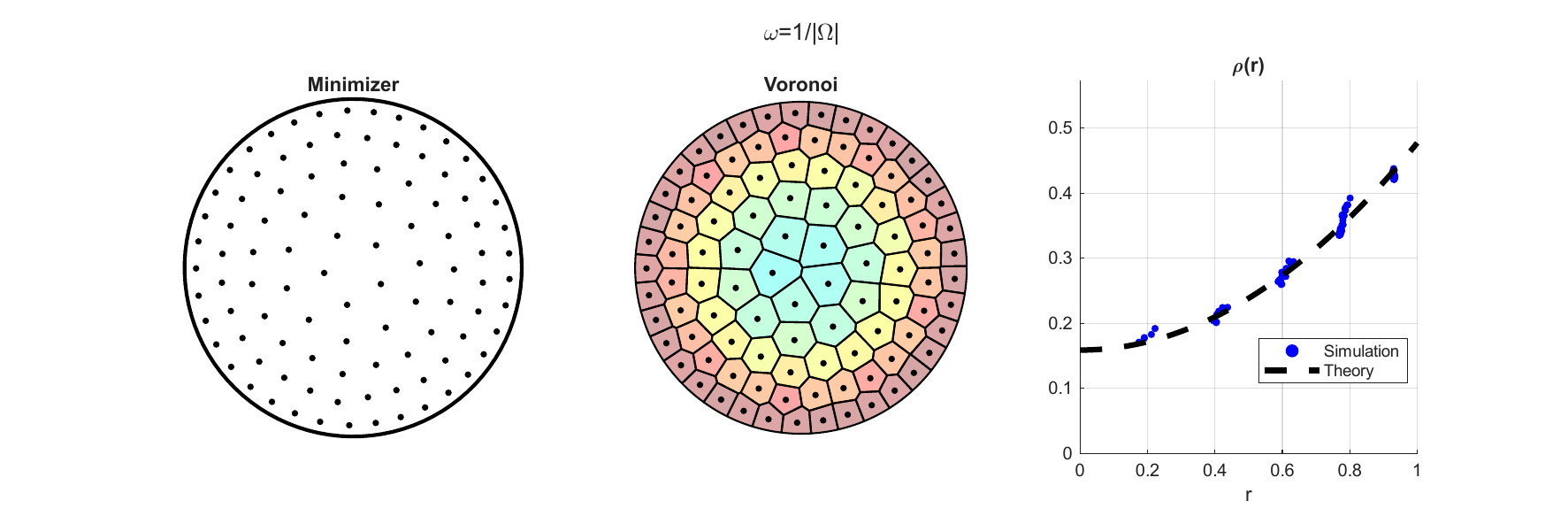} %
\includegraphics[width=0.95\textwidth]{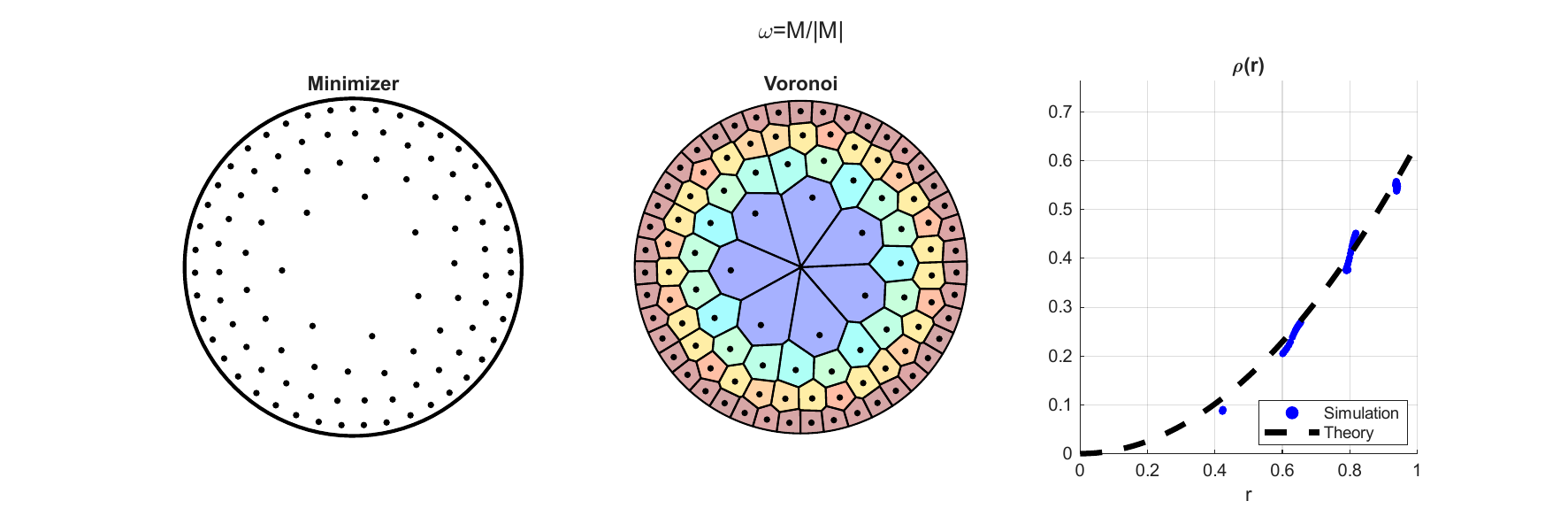}
\caption{ Optimal trap distribution with $N=100$ and $M(r)=r^{2}$ inside a
unit disk with (i)\ $\protect\omega=1/\left\vert \Omega\right\vert $ (top)\
and (ii)\ $\protect\omega =M/\left\vert M\right\vert $ (bottom). Left: trap
positions that minimize the Green's energy (\protect\ref{greenE}), computed
numerically (see text); Middle: the corresponding Voronoi cells. Right:
Comparison of numerical approximation of $\protect\rho(x)$ obtained from the
Voronoi diagram (see text) shown as \textquotedblleft
simulation\textquotedblright, with continuum limit (\protect\ref{230})\
shown as \textquotedblleft theory\textquotedblright.}
\label{fig:2d}
\end{figure}

\textbf{Disk. }When $\Omega$ is a unit disk, an explicit expression for the
Neumann Green's function is available \cite%
{kolokolnikov2009spot,kolokolnikov2003reduced},%
\begin{equation}
G(x,y)=-\frac{1}{2\pi}\log\left( \left\vert x-y\right\vert \right) +\frac {1%
}{4\pi}\left[ -\log(\left\vert x\right\vert ^{2}\left\vert y\right\vert
^{2}+1-2x\cdot y)+\left\vert x\right\vert ^{2}+\left\vert y\right\vert ^{2}-%
\frac{3}{2}\right] .
\end{equation}
As an example, we take the radial dependence for both $M(x)$ and $\omega(x)$%
: 
\begin{equation}
M(x)=a\left\vert x\right\vert ^{2}+b ,\qquad\omega(x)=\frac{c\left\vert
x\right\vert ^{2}+d}{\pi\left( c/2+d\right) } \,.  \label{diskAom}
\end{equation}
The effective trap density in the large-$N$ limit from (\ref{rhoavg}) is
then given by%
\begin{equation}
\rho(x)=\frac{1}{2\pi}\left[ \frac{a\left\vert x\right\vert ^{2}+b}{a/2+b} + 
\frac{c\left\vert x\right\vert ^{2} + d}{c/2+d}\right] .  \label{230}
\end{equation}
Solving (\ref{Vweq}) using radial coordinates, we find that 
\begin{align*}
V_{\mu+\omega}(x) & = -\left[ \frac{a}{4\pi\left( a+2b\right) }+\frac{c}{%
4\pi\left( c+2d\right) }\right] \left( \left\vert x\right\vert ^{2}-\frac{%
\left\vert x\right\vert ^{4}}{2}-\frac{1}{3}\right) , \\
\nabla G(x,y) & =-\frac{1}{2\pi}\frac{x-y}{\left\vert x-y\right\vert ^{2}}+%
\frac{1}{2\pi}\left[ -\frac{x\left\vert y\right\vert ^{2}-y}{\left\vert
x\right\vert ^{2}\left\vert y\right\vert ^{2}+1-2x\cdot y}+x\right] , \\
\nabla{V_{\mu+\omega}} & ={-\left[ \frac{a}{2\pi\left( a+2b\right) }+\frac{c%
}{2\pi\left( c+2d\right) }\right] \left( 1-\left\vert x\right\vert
^{2}\right) x} , \\
\nabla R(x,y) & =+\frac{1}{2\pi}\left[ -\frac{x\left\vert y\right\vert ^{2}-y%
}{\left\vert x\right\vert ^{2}\left\vert y\right\vert ^{2}+1-2x\cdot y}+x%
\right] ,\ \ \ \nabla R(x,x)=\frac{1}{2\pi}\left[ \frac{\left\vert
x\right\vert ^{2}-2}{\left\vert x\right\vert ^{2}-1}\right] x.
\end{align*}
With these explicit expressions at hand, it is easy to apply the gradient
descent method to ODE (\ref{ode}) to compute the optimal trap positions for
any given $a,b,c,d$. We illustrate this for $M(x)=|x|^{2}$ and for (i)
constant $\omega(x)=1/\pi$ and (ii)\ $\omega(x)=M/\left\vert M\right\vert
=|x|^{2} /(\pi/2).$ We used $N=100$ traps and applied the forward Euler
method to compute the steady state of the ODE (\ref{ode}). A few minutes on
a standard computer is enough to compute the steady state within $10^{-5}$,
measured as $\max_{k}\left\vert x_{k}^{\prime}(t)\right\vert $ at the final
time. The result is shown in Fig. \ref{fig:2d} (left side).

Once the steady state of (\ref{ode}) is computed, we construct Voronoi cells
(shown in the middle of Fig. \ref{fig:2d}) to estimate an effective density $%
\rho(x)$. We use $\rho(x_{j})\approx1/(NS_{j})$, where $S_{j}$ is the area
of the Voronoi cell around $x_{j}$. This is shown on the scatter plot on the
right, where $\left\vert x_{j}\right\vert $ is plotted versus the estimated $%
\rho(x_{j}),\ \ j=1\ldots N$. For comparison, this scatter plot is
superimposed with the explicit formula (\ref{230}) for the effective trap
density. Excellent agreement is observed. We emphasize that this is a
comparison between the minimization of the Green's energy and our
theoretical predictions; both relations are valid only in the limit $\nu\ll1$%
. An ultimate validation would require comparison with a numerical solution
of the original PDE (\ref{u2d}). However, as said earlier, the validity
condition $\nu\ll1$ would require dealing with extremely small $\varepsilon$%
. For this reason, we relegate such a comparison to \S \,\ref{sec:homog}, in
which the use of the homogenization technique will relax this condition.

\textbf{Rectangle. }Consider a rectangle $\Omega=(0,L)\times(0,H)$ and take
the heterogeneity to be one-dimensional and the weight to be uniform:%
\begin{equation}
M(x,y)=f(x),\qquad\omega=1/\left\vert \Omega\right\vert .  \label{recA}
\end{equation}
Then (\ref{rhoavg})\ yields 
\begin{equation}
\rho(x,y)=\frac{1}{2H}\left( \frac{f(x)}{\left\vert f\right\vert }+\frac {1}{%
L}\right) ,\qquad\left\vert f\right\vert :=\int_{0}^{L}f(x)dx,
\label{rectrhof}
\end{equation}
and 
\begin{equation}
\nabla{V_{\mu}(x,y)=\left( \frac{1}{H}\left[ \frac{F(x)}{\left\vert
f\right\vert }-\frac{x}{L}\right] ,0\right) ,\qquad
F(x):=\int_{0}^{x}f(s)\,ds;\qquad V_{\omega}=0,}  \label{recgradU}
\end{equation}
the latter because the starting distribution is uniform, so that $%
V_{\mu+\omega}=V_{\mu}$. For $f=ax+b$ this gives 
\begin{equation}
\nabla{V_{\mu}}(x,y)=\left( \frac{ax^{2}/2+bx}{\left( aL^{2}/2+bL\right) H}-%
\frac{x}{LH},\ 0\right) .  \label{recgradUab}
\end{equation}
The Neumann Green's function admits a rapidly convergent Fourier series
representation, with a resummation technique used to obtain a rapidly
convergent series for its regular part (see \cite{kolokolnikov2009spot},
Section 4.2 for these explicit formulas that we do not reproduce here). We
implemented these formulas in MATLAB and used them for computing the steady
state of the energy flow (\ref{ode}).

\begin{figure}[tb]
\begin{center}
\includegraphics[width=0.48\textwidth]{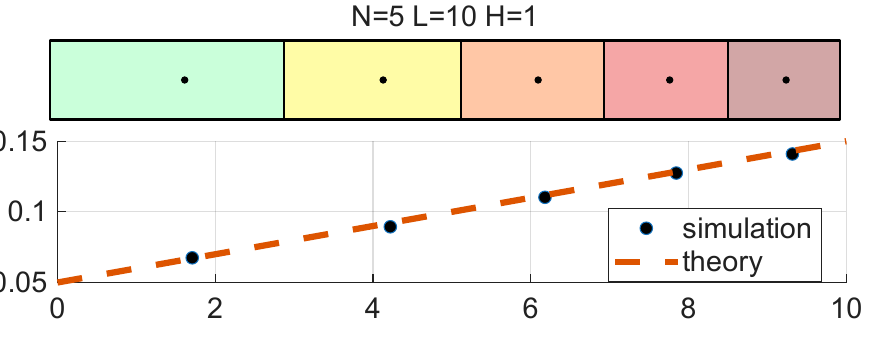} %
\includegraphics[width=0.48\textwidth]{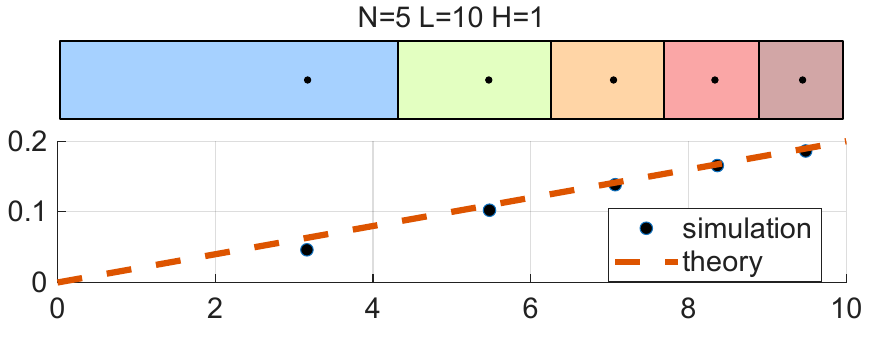}\\[0pt]
\includegraphics[width=0.48\textwidth]{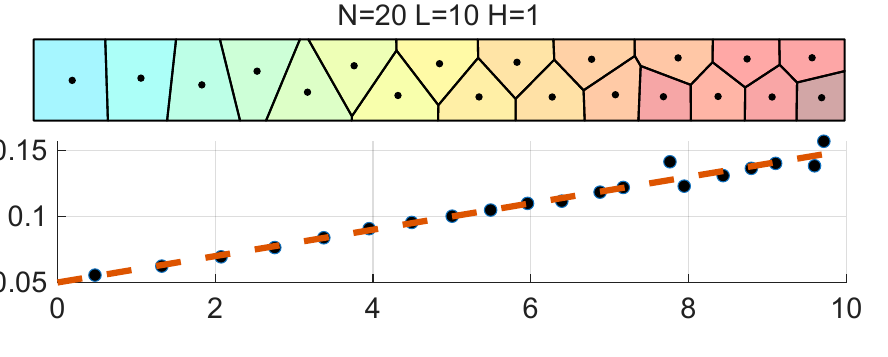} %
\includegraphics[width=0.48\textwidth]{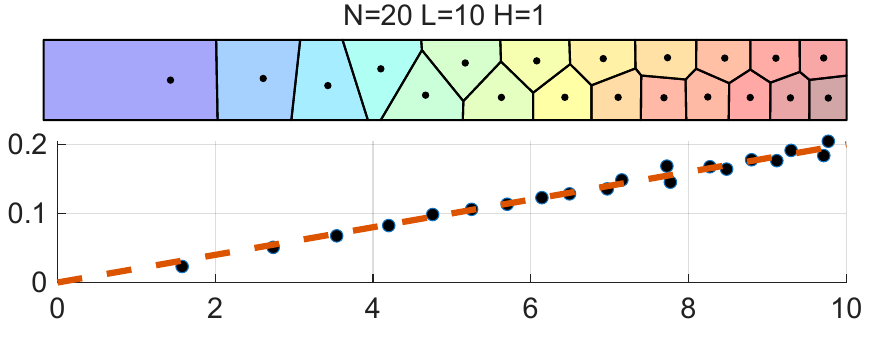}\\[0pt]
\includegraphics[width=0.48\textwidth]{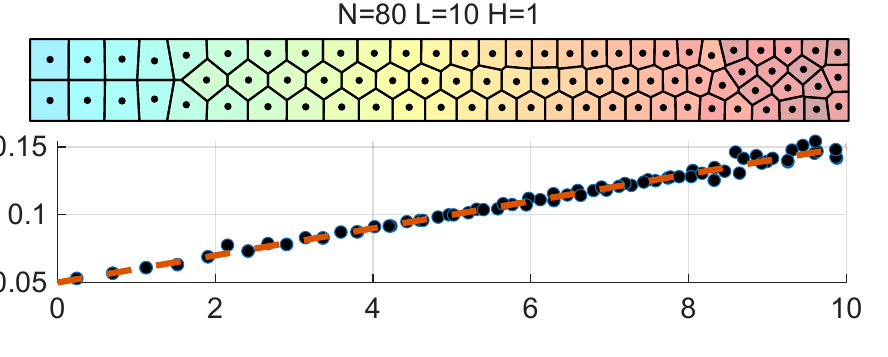} %
\includegraphics[width=0.48\textwidth]{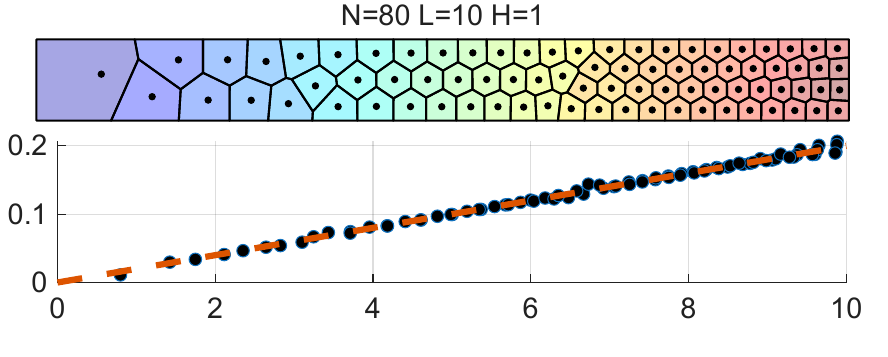}
\end{center}
\caption{ Optimal trap locations inside an elongated rectangle $%
(0,10)\times(0,1)$ with $M(x,y)=x$, for $N=5,20$ and $80$ traps, from top to
bottom. Left column: uniform starting distribution $\protect\omega%
\equiv1/\left\vert \Omega\right\vert $. Right column: $\protect\omega=%
\protect\mu=M/\left\vert M\right\vert $. In each panel the strip shows the
optimal trap positions together with their Voronoi cells, and the graph
below it compares the effective density obtained from the Voronoi cells
(dots) with the asymptotic prediction $\protect\rho(x)$ (dashed) from (%
\protect\ref{rhoavg}): $\protect\rho(x)=x/100+1/20$ for the left column and $%
\protect\rho =\protect\mu=x/50$ for the right. Note that the density is
strictly positive everywhere for uniform $\protect\omega$, but vanishes at $%
x=0$ for $\protect\omega=\protect\mu$.}
\label{fig:rect2}
\end{figure}

Figure \ref{fig:rect2} shows the results of a numerical experiment on a long
rectangle with $L=10$ and $H=1$. Here, we fixed $M(x,y)=x$ so that $%
\left\vert f\right\vert =50$ and (\ref{rectrhof}) gives $%
\rho(x,y)=x/100+1/20 $, and varied $N$, showing $N=5,20$ and $80$. Each row
of the figure shows the same $N$ for the two starting distributions: a
uniform $\omega= 1/|\Omega|$ on the left, and $\omega=\mu$ on the right, for
which (\ref{rhoavg}) collapses to $\rho=\mu=x/50$. It is interesting to note
that as $N$ is increased, the optimal configuration of traps changes from a
single line ($N=5$) to a mix of one and two rows ($N=20$), and so on. The
optimal configuration is denser on the right (lower diffusivity) and sparser
on the left (high diffusivity); note that (\ref{rectrhof}) predicts $%
\rho\rightarrow1/(2LH)>0$ rather than zero density as $x\rightarrow0$. The
variable density results in \textquotedblleft coexistence\textquotedblright\
of multi-row patterns (e.g. from one to four rows when $N=80$). Regardless
of pattern complexity, the overall effective density is in excellent
agreement with the theory, even for $N=5$, where the configuration is
essentially one-dimensional.

The two columns differ most near $x=0$, where diffusion is fastest. For a
uniform start the traps do not thin out indefinitely there: (\ref{rectrhof})
retains the floor $1/(2LH)$ contributed by $\omega$, and the leftmost trap
sits well inside the domain at every $N$. For $\omega=\mu$ that floor is
absent, $\rho$ vanishes linearly as $x\rightarrow0$, and the left end of the
rectangle is left progressively emptier --- there is no point placing traps
where the particle is unlikely to start. This is the case in which the weak-
and strong-trapping predictions agree for every $\kappa$, by (\ref%
{mainresult}), so the agreement seen in the right column does not by itself
test the $\kappa\ll1$ restriction; the left column does.

\section{Homogenization}

\label{sec:homog}

As the computation of $T_{2}$ in (\ref{T2})\ shows, the asymptotics we
derived in Lemma \ref{lemma:E2d}\ are only valid assuming that $\nu N\ll1.$
This is a particularly severe restriction in 2D\ (e.g. $%
\varepsilon=0.01,N=10 $ yields $\nu N\approx13.6$). In this section, we
develop an alternative approach based on homogenization \cite%
{cioranescu1982strange,marchenko2006homogenization,berezhkovskii2004boundary,berezhkovskii2006homogenization,muratov2008boundary,cheviakov2013narrow}%
. In the limit of many traps, the problem (\ref{u2d})\ can be homogenized as
follows. In 3D, consider the inner region near the trap $x_{k}.$ As in (\ref%
{1234y}), we rescale near $x_{k}$: $x=x_{k}+\varepsilon y;$ to leading order
we obtain $\Delta_{y}u\sim0,\ $with$\ u=0$ at $\left\vert y\right\vert =1.$
The solution is then given by {$u\sim b_{k}\bigl(1- \frac{1}{\left\vert
y\right\vert }\bigr) \sim b_k - \frac{b_{k}\varepsilon}{\left\vert x-x_{k}
\right\vert },$} where $b_{k}$ is to be determined. In the intermediate
regime $\varepsilon\ll\left\vert x-x_{k}\right\vert \ll1$ we then obtain {$%
u\rightarrow b_{k}$.}

In the outer region we then replace the trap at $x_{k}$ by a delta source
with strength {$4\pi b_{k}\varepsilon$} in the homogenized equation: 
\begin{align}
\Delta u & =-M(x)+\sum\limits_{j=1}^N 4\pi b_{j}\varepsilon\delta(x-x_{j}) 
\notag \\
& =-M(x)+4\pi\varepsilon\sum\limits_{j=1}^N u(x_{j})\delta(x-x_{j}).
\label{aeras}
\end{align}
Finally, recall that the density is given by $\rho(x)=\frac{1}{N}\sum_j
\delta(x-x_{j})$ so we replace $\sum_j u(x_{j})\delta(x-x_{j})\sim
u(x)N\rho(x)$ that leads to the homogenization limit (\ref{homog}), with $%
\nu $ given by (\ref{nu}). A similar derivation in 2D also yields (\ref%
{aeras})\ but with $4\pi\varepsilon$ replaced by $\frac{2\pi}{%
\log(1/\varepsilon)}.$

The minimization problem is to find the optimal $\rho(x)$ subject to $%
\int_\Omega\rho dx=1$ such that the $\omega$-averaged MFPT $\overline{u}%
_{\omega} =\int_{\Omega}\omega u$ of (\ref{ubarw}) is minimized. In general,
there is no explicit solution, and the optimal density $\rho\left(
x\right) $ has to be found numerically. 
However the regimes $\kappa\ll1$ and $\kappa\gg1$ can be computed explicitly
for an arbitrary $\omega$, as we now show.

\textbf{Weak trapping limit }$\kappa\ll1$\textbf{.} Note that this
corresponds precisely to the validity region $\nu N\ll1$ of Lemma \ref%
{lemma:E2d}; we now rederive this result via homogenization approach.

We expand $u$ in powers of $\kappa$ as follows%
\begin{equation*}
u=\frac{1}{\kappa}u_{0}+u_{1}+\kappa u_{2}+\ldots
\end{equation*}
so that 
\begin{equation*}
\Delta u_{0}=0,\ \ \ \ \ \Delta u_{1}=\rho u_{0}-M(x),\ \ \ \Delta
u_{2}=\rho u_{1},
\end{equation*}
subject to the boundary conditions\ $\partial_{n}u_{i}=0$ on $\partial\Omega$%
, with $i=0,1,2\ldots.$ Then $u_{0}$ is constant and integrating the
equation for $u_{1}$ we obtain $u_{0}=\left\vert M\right\vert .$ The Green's
second identity yields 
\begin{align*}
u_{1}(y) & =-\left\vert M\right\vert \int_{\Omega}G(x,y)\left(
\rho(x)-\mu(x)\right) dx+K \\
& =\left\vert M\right\vert V_{\rho-\mu}(y)+K ,
\end{align*}
where $V_{\psi}$ is defined in (\ref{Vw})$\ $and $K$ is a constant that can
be determined from the equation for $u_{2}$, yielding $\int_{\Omega}\rho
u_{1}=0,$ so that 
\begin{equation*}
u_{1}(y)=\left\vert M\right\vert \left( V_{\rho-\mu}(y)-\int_{\Omega}\rho
V_{\rho-\mu}\right) .
\end{equation*}
We therefore obtain the two-term expansion 
\begin{equation}
\overline{u}_{\omega}=\frac{\left\vert M\right\vert }{\kappa}+\left\vert
M\right\vert \int_{\Omega}(\omega-\rho)V_{\rho-\mu}\, +O(\kappa).
\label{ubarwhomog}
\end{equation}

The optimal $\rho(x)$ minimizes (\ref{ubarwhomog}) subject to $\int_{\Omega
}\rho=1.$ Taking a variational derivative and applying a Lagrange
multiplier, the optimal density satisfies%
\begin{equation*}
V_{\omega}-2V_{\rho}+V_{\mu}=\lambda.
\end{equation*}
Applying $\Delta$ to both sides yields $\omega-2\rho+\mu=0$ and recovers (%
\ref{rhoavg}). Finally, substituting $\rho=\frac{1}{2}\left( \mu
+\omega\right) $ into (\ref{ubarwhomog}), we obtain 
\begin{equation*}
\overline{u}_{\omega}=\frac{\left\vert M\right\vert }{\kappa}+\frac{%
\left\vert M\right\vert }{4}\int_{\Omega}(\omega-\mu)V_{\omega-\mu}\,{%
+O(\kappa).}
\end{equation*}
The integral can be further simplified using the identity $\int_{\Omega}\psi
V_{\psi}=-\int_{\Omega}\left\vert \nabla V_{\psi}\right\vert ^{2}$, valid
for any mean-zero $\psi$ (multiply $\Delta V_{\psi}=\psi$ by $V_{\psi}$ and
integrate by parts, the boundary term vanishing by $\partial_{n}V_{\psi}=0$%
), resulting in (\ref{ubarsmallk}). Taking $\omega=\mu$ we get $\rho=\mu$
and $\overline{u}_{\omega}=\frac{\left\vert M\right\vert }{\kappa}$,
recovering (\ref{mainresult})\ for the special case of small $\kappa=\nu N.$

\textbf{Strong trapping limit }$\kappa\gg1$. Here, we expand $u=\frac {1}{%
\kappa}u_{0}+\ldots.$ At the leading order we then obtain $u_{0}=\frac {M}{%
\rho},$ so that $\overline{u}_{\omega}\sim\frac{1}{\kappa}\int_{\Omega }%
\frac{\omega M}{\rho}.$ Minimization under the constraint $\int_\Omega
\rho=1 $ using a Lagrange multiplier yields $\frac{\omega M}{\rho^{2}}%
=\lambda,$ so that 
\begin{equation}
\rho=\frac{\sqrt{\omega M}}{\int_{\Omega}\sqrt{\omega M}} \,,  \label{rhobig}
\end{equation}
which is identical with (\ref{rhogeom}); substituting (\ref{rhobig}) into $%
\overline {u}_{\omega}\sim\kappa^{-1}\int_{\Omega}\omega M/\rho$\ yields (%
\ref{ubarwklarge}). Choosing $\omega=\mu$ again gives $\rho=\mu$, $\overline{%
u}_{\omega}\sim\frac{1}{\kappa}\left\vert M\right\vert $, recovering (\ref%
{mainresult})\ for the special case of large $\kappa=\nu N.$

\textbf{Arbitrary trapping}. For general $\kappa$ the optimal density has no
closed form, but it is characterized by a simple optimality system. We seek
to minimize $J\left[ \rho\right] =\int_{\Omega}\omega u.$ Perturbing $%
\rho\rightarrow\rho+{\delta\rho}$ and $u\rightarrow u+\delta u$ yields $%
\delta J=\int_{\Omega}\omega\delta u\ $which we need to rewrite in terms of $%
\delta\rho$. We have $\Delta\delta u-\kappa\rho\,\delta u=\kappa
u\,\delta\rho.$ Introduce an adjoint $v$ to satisfy $\Delta v-\kappa\rho
v=-\omega$ in $\Omega$, with Neumann boundary condition on $\partial\Omega$.
Multiplying this equation by $\delta u$ and integrating by parts we obtain%
\begin{equation*}
{-}\int_{\Omega}\omega\delta u=\int_{\Omega}\delta u\left( \Delta
v-\kappa\rho v\right) =\int_{\Omega}v\left( \Delta\delta u-\kappa
\rho\,\delta u\right) =\int_{\Omega}v\kappa u\,\delta\rho.
\end{equation*}
This yields $\delta J/\delta\rho={-}\kappa uv.$ Finally we introduce a
Lagrange multiplier for the constraint $\int_{\Omega}\rho=1$ so that
stationarity of $\overline{u}_{\omega}$ requires $\kappa uv=\lambda.$ In
summary, the optimal density $\rho$ satisfies 
\begin{equation}
\left\{ 
\begin{array}{l}
\Delta u-\kappa\rho u=-M\quad\text{in }\Omega,\qquad\partial_{n}u=0\ \text{%
on }\partial\Omega, \\ 
\Delta v-\kappa\rho v=-\omega\quad\text{in }\Omega,\qquad\partial _{n}v=0\ 
\text{on }\partial\Omega, \\ 
\,u\,v=\func{const},\qquad\int_{\Omega}\rho=1.%
\end{array}
\right.  \label{kappa-any}
\end{equation}
In particular, if $\omega=\mu$, then $v\propto u$, and the optimality
condition reduces to $u=\mathrm{const}$: the optimal density equalizes the
MFPT over all starting points, and (\ref{homog}) gives $\rho=\mu$ together
with $u\equiv\left\vert M\right\vert /\kappa$, for every $\kappa$, which is
precisely (\ref{mainresult}). For a homogeneous medium ($M$ constant, $%
\omega $ uniform) this says that the optimal density is uniform, $%
\rho=1/\left\vert \Omega\right\vert $.

\subsection{Proofs of Theorems \protect\ref{thm:alpha}--\protect\ref%
{thm:general}}

\label{sec:proofs}

We now give the proofs of Theorems \ref{thm:alpha} and \ref{thm:general};
Theorem \ref{thm:minimax} follows from Theorem \ref{thm:alpha} as noted in
the Introduction.

\textbf{Proof of Theorem \ref{thm:alpha}.} Multiply (\ref{homog})\ by $u$
and integrate to obtain 
\begin{align*}
\left\vert M\right\vert \int_{\Omega}\mu u & =\kappa\int_{\Omega}\rho
u^{2}-\int_{\Omega}u\Delta u \\
& =\kappa\int_{\Omega}\rho u^{2}+\int_{\Omega}\left\vert \nabla u\right\vert
^{2} \\
& \geq\kappa\int_{\Omega}\rho u^{2}.
\end{align*}
By Cauchy--Schwarz, $\int_{\Omega}\rho u^{2}\geq\left( \int_{\Omega}\rho
u\right) ^{2}$ (recall that $\rho\geq0$ and $\int_{\Omega}\rho=1$), and
integrating (\ref{homog}) we obtain $\int_{\Omega}\rho u=\left\vert
M\right\vert /\kappa$. Putting these together yields%
\begin{equation*}
\int_{\Omega}\mu u\geq\frac{\left\vert M\right\vert }{\kappa}.
\end{equation*}
Equality requires both $\int_{\Omega}\left\vert \nabla u\right\vert ^{2}=0$
and $\int_{\Omega}\rho u^{2}=\left( \int_{\Omega}\rho u\right) ^{2}.$ The
former implies that $u$ is constant, which in turn implies the latter.
Substituting a constant $u$ into (\ref{homog})\ gives $\rho\propto\mu$,
hence $\rho=\mu$ by normalization, and then $u=\left\vert M\right\vert
/\kappa$. Conversely, for $\rho=\mu$ the constant $u\equiv\left\vert
M\right\vert /\kappa$ solves (\ref{homog}), so equality holds. $\blacksquare$

\textbf{Proof of Theorem \ref{thm:general}.} Since $M>0$ and $\rho \geq 0$,
the maximum principle gives $u>0$ in $\Omega $. Dividing (\ref{homog})\ by $%
u $ and integrating by parts, using $\int_{\Omega }\rho =1$, we obtain 
\begin{equation}
\int_{\Omega }\frac{M}{u}=\kappa -\int_{\Omega }\frac{\Delta u}{u}=\kappa
-\int_{\Omega }\left\vert \frac{\nabla u}{u}\right\vert ^{2}\leq \kappa .
\label{eq:Aoveru}
\end{equation}%
Next, use the Cauchy--Schwarz inequality in the form $\int a^{2}\geq \left(
\int ab\right) ^{2}/\int b^{2}$. Choosing $a=\sqrt{\omega u}$ and $b=\sqrt{%
M/u}$ yields 
\begin{equation}
\int_{\Omega }\omega u\geq \frac{\left( \int_{\Omega }\sqrt{\omega M}\right)
^{2}}{\int_{\Omega }M/u}.  \label{eq:cauchyC}
\end{equation}%
From (\ref{eq:Aoveru})\ the right-hand side is at least $\left( \int_{\Omega
}\sqrt{\omega M}\right) ^{2}/\kappa =\frac{\left\vert M\right\vert }{\kappa }%
\left( \int_{\Omega }\sqrt{\omega \mu }\right) ^{2}$, which yields the
inequality (\ref{eq:thmgeneral}). This inequality becomes an equality when
we take $\omega =\rho =\mu $ (since in this case the solution is $%
u=\left\vert M\right\vert /\kappa $). Conversely, suppose that $\int_{\Omega
}\omega u=\frac{\left\vert M\right\vert }{\kappa }\left( \int_{\Omega }\sqrt{%
\omega \mu }\right) ^{2}.$ This is only possible if the inequalities (\ref%
{eq:Aoveru})\ and (\ref{eq:cauchyC}) both become equalities. The former
implies $\nabla u=0$, so that $u$ is constant; the latter is the
Cauchy--Schwarz inequality, which becomes an equality only if $\sqrt{\omega u%
}=C\sqrt{\left\vert M\right\vert \mu /u}$ for some constant $C>0$. With $u $
constant this forces $\omega =C_{2}\mu $ for some constant $C_{2};$ since {$%
\mu $ and }$\omega $ are{\ both densities}, this implies $\omega =\mu $; and
then, exactly as in the proof of Theorem \ref{thm:alpha}, a constant $u$ in (%
\ref{homog}) forces $\rho =\mu $. $\blacksquare $

\subsection{Numerics}

\label{sec:homognum}

\begin{figure}[ptb]
\begin{center}
\includegraphics[width=0.99\textwidth]{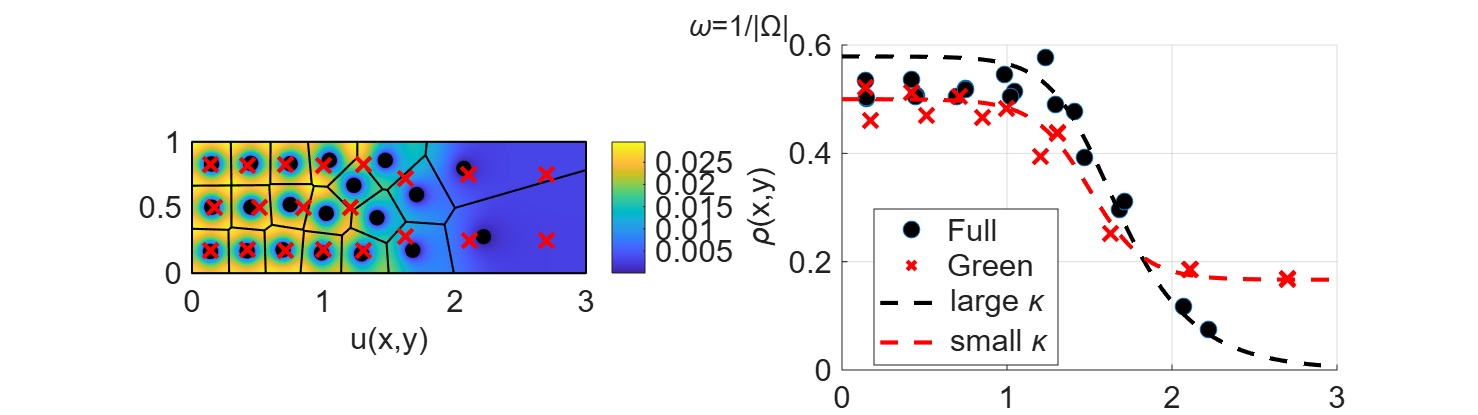} %
\includegraphics[width=0.99\textwidth]{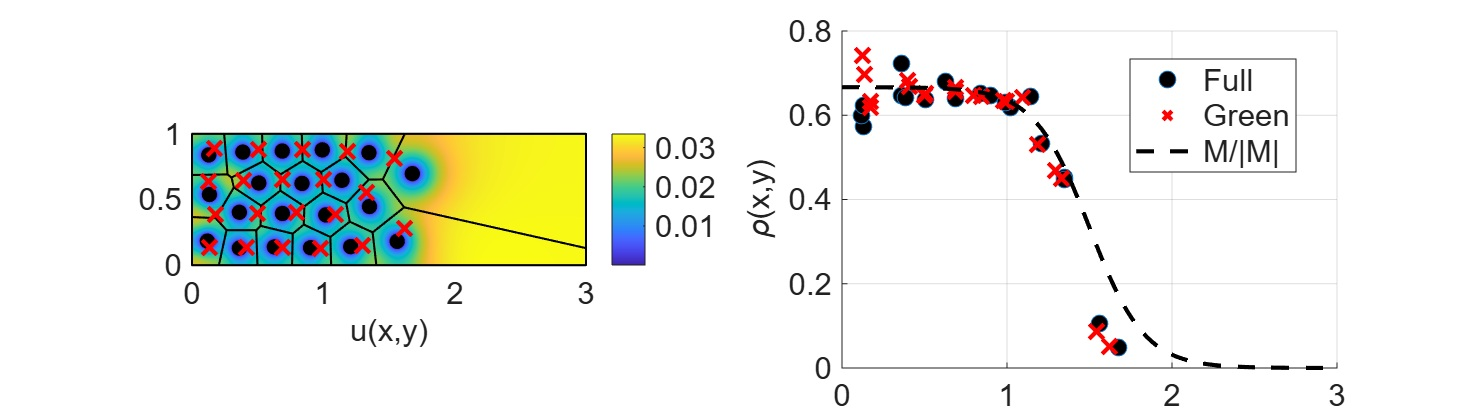}
\end{center}
\caption{ Top: Optimal trap locations inside a rectangle $(0,3)\times(0,1)$,
with $M(x,y)=\tanh(3(1.5-x))+1$ (slow on the left, fast on the right), $%
\protect\omega=1/\left\vert \Omega\right\vert ,$ and with $N=20,\protect%
\varepsilon=0.05$ (so that $\protect\kappa=\protect\nu N\approx42$). Black
dots ('full') correspond to minimizing the full problem (\protect\ref{u2d})
numerically. Red crosses ('Green') correspond to the minimization of the
Green's energy (\protect\ref{greenE}). Top left:\ color plot shows the
corresponding solution $u(x,y)$ of the full problem {(\protect\ref{u2d})}
with optimal trap locations, and the corresponding Voronoi cells. Optimal
trap locations computed via Green's energy minimization are also shown. Top
right: effective trap density obtained from the two computed configurations,
compared with the small-$\protect\kappa$ prediction (\protect\ref{rhoavg})
and the large-$\protect\kappa$ prediction (\protect\ref{rhobig}). Since $%
\protect\kappa\approx42$ here, the large-$\protect\kappa$ formula is the
relevant one, and it is the one the full numerics follows. Bottom:\ same
parameters as top, except $\protect\omega=M/\left\vert M\right\vert .$ Both
full and Green's function optimizations approximate the continuum density $%
\protect\rho=M/\left\vert M\right\vert .$}
\label{fig:kappa}
\end{figure}

\textbf{Rectangle in 2D. }To test the homogenization theory against the
original problem (\ref{u2d}), we used Claude AI\ to implement a 2D numerical
solver as well as a trap location optimizer based on the original full
problem (\ref{u2d}). Details of the implementation 
are given in Appendix \ref{app:numerics}.

Figure \ref{fig:kappa}(top) illustrates the two regimes for $\Omega
=(0,3)\times(0,1)$ with $M(x,y)=\tanh(3(1.5-x))+1$, $N={20}$ and $%
\varepsilon=0.05$, with $\omega\equiv1/\left\vert \Omega\right\vert $, a
uniform starting distribution. Then $\kappa=\nu N\approx{42}$, corresponding
to the strong-trapping regime. As expected, the full numerics follow the
strong-trapping limit (\ref{rhobig}) rather than the weak-trapping limit (%
\ref{rhoavg}), which corresponds to the Green's energy. This clearly
illustrates that homogenization theory is usually a superior choice to the
Green's function energy (Lemma \ref{lemma:E2d}) for most parameter values in
2D.

In Figure \ref{fig:kappa}(bottom), we use the same parameters except that we
start with the equilibrium distribution $\omega=\mu.$ In this case, both
strong- and weak-trapping limits coincide; both yield $\rho =\mu.$
Correspondingly, both the full solver (black dots) and the Green's energy
optimizer (red crosses) converge to the $\rho=\mu$ curve as expected. {%
Note that the shown solution $u(x,y)$ for $N = 20$ is not yet
constant over the domain, as it should be in the limit $N\to \infty$.}

\begin{figure}[tbp]
\begin{center}
\includegraphics[width=0.99\textwidth]{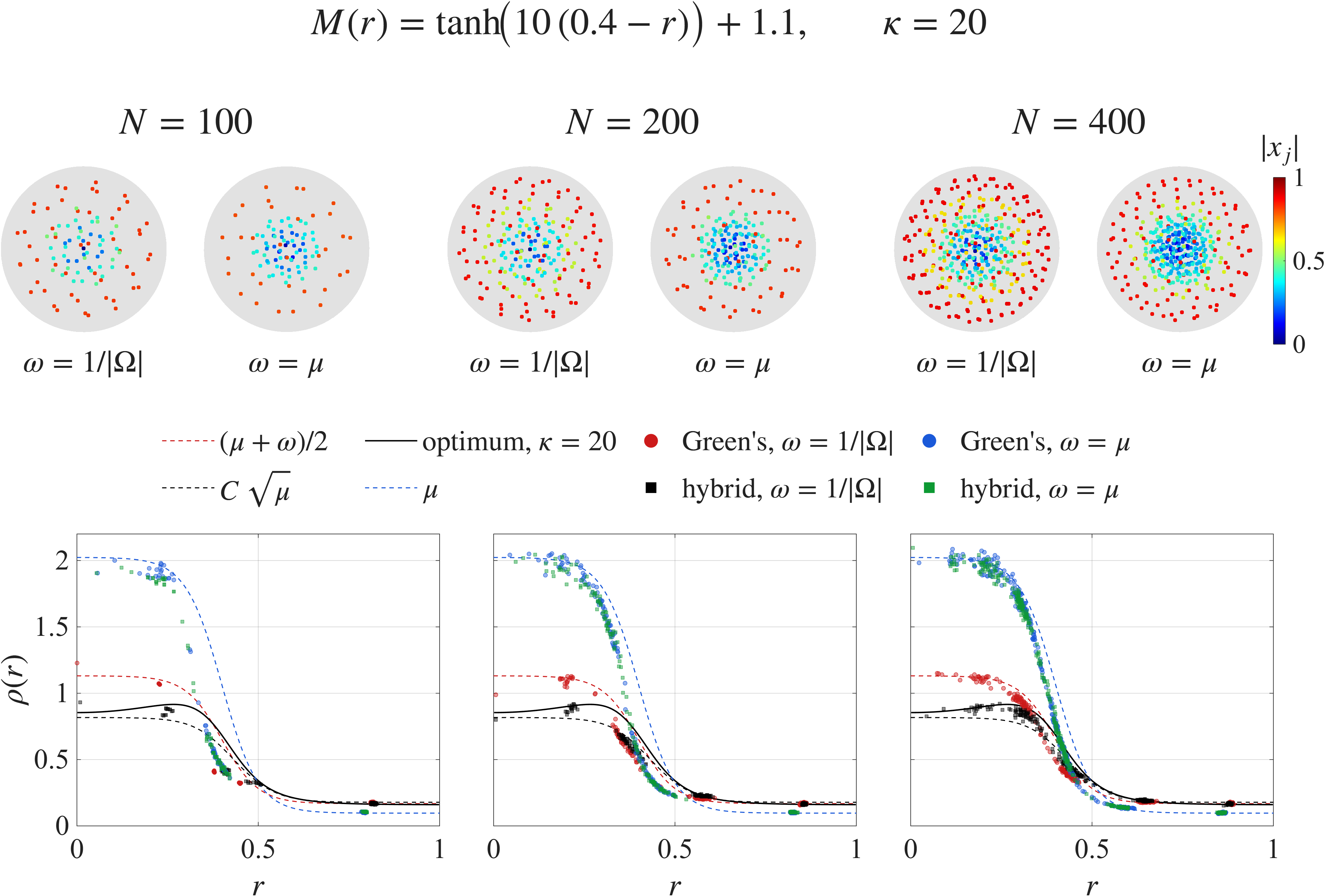}
\end{center}
\caption{ Optimal trap configurations in the unit ball for $M(r)=\tanh
(10(0.4-r))+1.1$ (a slow core inside a fast shell, with an interface of
width $\sim 0.1$ at $r=0.4$), with $N=100$, $200$ and $400$ traps (left to
right), all at $\protect\kappa =\protect\nu N=20$, i.e., the strong-trapping
regime. Top: the minimizers of the hybrid objective (\protect\ref{fullubar})
for the uniform starting distribution $\protect\omega \equiv 1/\left\vert
\Omega \right\vert $ and for $\protect\omega =\protect\mu $, colored
according to $\left\vert x_{j}\right\vert $. Bottom: the effective densities 
$\protect\rho (x_{j})\approx 1/(N\left\vert \mathcal{V}_{j}\right\vert )$
from the Voronoi cells, plotted against $\left\vert x_{j}\right\vert $, for
the minimizers of the Green's energy (\protect\ref{greenE}) (circles) and of
the hybrid objective (\protect\ref{fullubar}) (squares), for $\protect\omega %
\equiv 1/\left\vert \Omega \right\vert $ (red and black) and for $\protect%
\omega =\protect\mu $ (blue and green). For $\protect\omega \equiv
1/\left\vert \Omega \right\vert $ the curves are: the small-$\protect\kappa $
prediction (\protect\ref{rhoavg}), $(\protect\mu +\protect\omega )/2$ (red
dashed); the large-$\protect\kappa $ prediction (\protect\ref{rhobig}), here 
$\propto \protect\sqrt{\protect\mu }$ (black dashed); and the optimum of the
optimality system (\protect\ref{kappa-any}) at $\protect\kappa =20$ (black
solid); for $\protect\omega =\protect\mu $ all predictions coincide at $%
\protect\rho =\protect\mu $ (blue dashed), by (\protect\ref{mainresult}).
The Green's-energy minimizer follows the small-$\protect\kappa $ curve even
though $\protect\kappa =20$, while the hybrid solution follows the finite-$%
\protect\kappa $ optimum, which at this $\protect\kappa $ is already close
to the large-$\protect\kappa $ limit; for $\protect\omega =\protect\mu $ the
two objectives give indistinguishable configurations and densities. The
interface at $r=0.4$ is thinner than the trap spacing at $N=100$, which is
why the measured densities lag the predictions there; the lag disappears as $%
N$ grows and the core is resolved by several shells of traps. }
\label{fig:sphere3d}
\end{figure}

\textbf{Unit sphere in 3D.} For a 3D unit ball, an explicit formula for the
Neumann Green's function is available \cite{cheviakov2011optimizing}:

\begin{equation}
G(x,y)=\frac{1}{4\pi}\left[ \frac{1}{\left\vert x-y\right\vert }+\frac {1}{%
\sqrt{1-2x\cdot y+\left\vert x\right\vert ^{2}\left\vert y\right\vert ^{2}}}%
+\log\frac{2}{1-x\cdot y+\sqrt{1-2x\cdot y+\left\vert x\right\vert
^{2}\left\vert y\right\vert ^{2}}}+\frac{\left\vert x\right\vert
^{2}+\left\vert y\right\vert ^{2}}{2}-\frac{14}{5}\right] .  \label{Gball}
\end{equation}
The regular part (\ref{R}) is 
\begin{equation}
R(x,x)=\frac{1}{4\pi}\left[ \frac{1}{1-|x|^{2}}-\log(1-|x|^{2})+|x|^{2}-%
\frac {14}{5}\right] .  \label{Rball}
\end{equation}
This makes it easy to optimize the Green's energy formulation. {\ However,
solving the full optimization problem with relatively large $N$ (say, few
hundred) is still a difficult task. For this reason,} we use the hybrid
formulation (\ref{140}), rather than its two-term series expansion as given
by Lemma \ref{lemma:E2d}. In other words, we solve a linear system of $N+1$
equations to get the unknowns $T$ and $C_{1},\ldots,C_{N}$ that determine $%
\overline{u}_{\omega}$ via (\ref{fullubar}). 
We then minimize (\ref{fullubar}%
) over the trap positions using gradient descent with Barzilai--Borwein step
lengths and an analytic gradient obtained from the adjoint, as described in
Appendix \ref{app:numerics}. The Green's-energy configurations shown for
comparison are produced by the same minimization applied to $E$ of (\ref%
{greenE}), and the effective densities are estimated from Voronoi cell
volumes as in 2D, $\rho(x_{j})\approx1/(N\left\vert \mathcal{V}%
_{j}\right\vert )$ with $\left\vert \mathcal{V}_{j}\right\vert $ being the
volume of the Voronoi cell of $x_{j}$ clipped to the unit ball.

We validated
the hybrid system against a full numerical solve of (\ref{u2d}) with
$\varepsilon=0.04, N=20$, requiring 8 million cells.
The two agreed to within less than $1\%$-- see Appendix
\ref{app:numerics}. However we have used the hybrid system only for
optimization and comparison with continuum limit.

Figure \ref{fig:sphere3d} illustrates the results for 
$M(r)=\tanh \left( 10\left( 0.4-r\right) \right) +1.1,$ 
which models a slow core surrounded by a fast shell, with a sharp interface
of width about $0.1$ at $r=0.4$. The three columns correspond to $N=100$, $%
200$ and $400$ traps, all at the same trapping strength $\kappa =20$ (so
that $\varepsilon $ is halved from one column to the next). For a uniform
starting distribution the two objectives give visibly different
configurations: the Green's-energy minimizer (valid only when $\nu N\ll 1$)
follows the small-$\kappa $ prediction (\ref{rhoavg}), whereas the minimizer
of the hybrid system (\ref{fullubar}) follows the finite-$\kappa $ optimum
of (\ref{kappa-any}), which at $\kappa =20$ already lies close to the large-$%
\kappa $ limit (\ref{rhobig}). For $\omega =\mu $ the two objectives produce
the same configuration, with density $\rho =\mu $, as (\ref{mainresult})
predicts for every $\kappa $.

The agreement between the discrete configurations and the continuum
densities is spatially non-uniform, and depends on how the interface width
compares with the distance between neighbouring traps. At $N=100$ the
interface is under-resolved, resulting in a significant leftward shift of
the trap density near the interface. The agreement is visibly improved as $N$
grows. Similar improvement can be observed by making the interface less
sharp (not shown). We invite the reader to experiment with this effect in 2D
using the web app \cite{webapp}.

\section{One-dimensional setting}

\label{sec:1D}

We now consider several MFPT problems in one dimension, including
inhomogeneous diffusion, advection, and a one-dimensional approximation of a
thin-domain problem. Unlike the two-dimensional case, a one-dimensional trap
splits the interval into two subintervals that can be solved independently,
and then combined to compute the overall MFPT. This gives a significant
simplification for the analysis. In fact, all of the problems we consider
will reduce to minimizing the \textquotedblleft energy\textquotedblright\ of
the general form 
\begin{equation}
E(x_{1},x_{2},\ldots x_{N})=\sum_{k=0}^{N}f(l_{k},x_{k}),\quad
l_{k}=x_{k+1}-x_{k},  \label{4:05}
\end{equation}
inside an interval $[0,L]$ with $0=x_{0}<x_{1}<\ldots<x_{N}<x_{N+1}=L$ (for
simplicity, we assumed here that the boundaries $x_{0}=0$ and $x_{N+1}=L$
are also trapping). Here, $f(l,x)$ is some function that depends on the
specific problem and approximates the contribution of an interval of length $%
l$ located near $x$ to the objective (here, the MFPT), weighted by the
starting distribution when the latter is non-uniform. We are interested in
computing the minimizer $x_{1},\ldots,x_{N}$ in the limit $%
N\rightarrow\infty $. In this limit, define a continuous function $\ell(x)$
such that $l_{k}\rightarrow\ell(x_{k})$ as $N\rightarrow\infty.$

\begin{prop}
\label{prop:cont} Let $(x_{1},x_{2},\ldots ,x_{N})$ be the minimizer of (\ref%
{4:05}), i.e., 
\begin{equation}
(x_{1},x_{2},\ldots ,x_{N})=\argmin_{0<x_{1}^{\prime }<\ldots
<x_{N}^{\prime }<L}\{E(x_{1}^{\prime },x_{2}^{\prime },\ldots ,x_{N}^{\prime
})\}.  \label{qoa}
\end{equation}%
Let $\ell (x)$ solve%
\begin{equation}
f(\ell ,x)-\ell f_{\ell }(\ell ,x)=\lambda ,  \label{1854}
\end{equation}%
where the constant $\lambda $ (a Lagrange multiplier) is determined by the
integral condition%
\begin{equation}
N=\int_{0}^{L}\frac{1}{\ell (x)}dx\,.  \label{1855}
\end{equation}%
Here and below, subscripts on $f$ denote partial derivatives with respect to
its first argument $\ell $ and its second argument $x$; $l_{k}$ continues to
denote the discrete gaps and $\ell (x)$ their continuum limit. Then {\ $%
x_{j+1}-x_{j}\sim \ell (x_{j})$ }for $N\gg 1.$
\end{prop}

\textbf{Proof.} We pass to the continuum directly. Since $dx/\ell
(x_{k})\sim dk$, the sum (\ref{4:05}) becomes%
\begin{equation*}
E=\sum_{k=0}^{N}f(l_{k},x_{k})\sim \int_{0}^{N}f(\ell (x_{k}),x_{k})\,dk\sim
\int_{0}^{L}f(\ell (x),x)\frac{dx}{\ell (x)},
\end{equation*}%
while the number of traps is $\int_{0}^{N}dk=\int_{0}^{L}\frac{dx}{\ell (x)}%
=N.$ The large-$N$ limit of (\ref{4:05}) is therefore the constrained
variational problem%
\begin{equation}
\min_{\ell (x)}\int_{0}^{L}f(\ell (x),x)\frac{dx}{\ell (x)}\text{ \ subject
to \ }\int_{0}^{L}\frac{dx}{\ell (x)}=N,\ \ \ell (x)>0.  \label{varprob}
\end{equation}%
Introducing a Lagrange multiplier $\lambda $ for the constraint, this is
equivalent to%
\begin{equation*}
\min_{\ell (x),\ \lambda }\int_{0}^{L}\frac{f(\ell (x),x)-\lambda }{\ell (x)}%
dx.
\end{equation*}%
The integrand depends on $\ell $ pointwise and not on its derivatives, so
the variational derivative is an ordinary derivative,%
\begin{equation*}
\frac{\partial }{\partial \ell }\left( \frac{f(\ell ,x)-\lambda }{\ell }%
\right) =0\iff \ell f_{\ell }-\left( f-\lambda \right) =0,
\end{equation*}%
which is (\ref{1854}); the constraint in (\ref{varprob}) is (\ref{1855}). $%
\blacksquare $

We now investigate several problems by using this framework.

\subsection{Heterogeneous diffusion}

\label{sec:D}

Consider a Brownian particle governed by a stochastic ODE $dx_{t} =\sqrt{%
2D(x_{t})}dW_{t},$ where $D(x) > 0$ is a position-dependent diffusivity. The
corresponding Fokker--Planck equation is 
\begin{equation}
p_{t}=\left( D^{\theta}\left( D^{1-\theta}p\right) _{x}\right) _{x},
\label{lawley}
\end{equation}
where the parameter $0\leq\theta\leq1$ depends on the stochastic
interpretation, see \cite%
{vanKampen1981,Sokolov2010,tunglawley2025escape,tunglawley2026stochastic}
for details. In particular, the It\^{o} interpretation corresponds to $%
\theta=0$, whereas the Stratonovich interpretation corresponds to $%
\theta=1/2 $.

Assuming $0=x_{0}<x_{1}<\ldots<x_{N+1}=L$, the MFPT\ problem with $N$ traps
reads%
\begin{equation}
D^{1-\theta}\left( D^{\theta}u_{x}\right) _{x}=-1 \quad\text{on}~
(x_{j},x_{j+1}), \qquad u(x_{j})=0 \quad(j=0,\ldots, N),  \label{1733}
\end{equation}
and, as in 2D and 3D, we are interested in minimizing the $\omega$-averaged
MFPT 
\begin{equation}
\overline{u}_{\omega}=\int_{0}^{L}\omega(x)u(x)\,dx,\qquad\int_{0}^{L}%
\omega=1,\ \ \omega\geq0,  \label{ubarw1D}
\end{equation}
as a function of trap locations $x_{j}$; the uniform choice $\omega\equiv1/L$
recovers the global MFPT $T=\int_{0}^{L}u(x)\,dx/L$. Our goal is to find the
optimal density distribution of the trap locations as a function of $D(x)$
and $\omega(x)$, in the limit of large $N$.

The ODE (\ref{1733}) for each interval $(x_{j},x_{j+1})$ can be easily
solved in terms of quadratures: 
\begin{equation}
u(x) = C_{j} \int_{x_{j}}^{x} D^{-\theta}(y) dy - \int_{x_{j}}^{x} dy_{1}
D^{-\theta}(y_{1}) \int_{x_{j}}^{y_{1}} dy_{2} \, D^{\theta-1}(y_{2}) \qquad
x\in(x_j,x_{j+1}),
\end{equation}
where the constant 
\begin{equation}
C_{j} = \left( \int_{x_{j}}^{x_{j+1}} D^{-\theta}(y) dy\right) ^{-1}
\int_{x_{j}}^{x_{j+1}} dy_{1} D^{-\theta}(y_{1}) \int_{x_{j}}^{y_{1}} dy_{2}
\, D^{\theta-1}(y_{2})
\end{equation}
ensures the boundary condition $u(x_{j+1}) = 0$. In the following, we deduce
more explicit asymptotic results, assuming that $l_{k}=x_{k+1}-x_{k}$ is
small.

Let $x=x_{k}+l_{k}y$ and $u(x)=l_{k}^{2}U(y)$. Then at leading order we
obtain 
\begin{equation*}
D_{0}U_{0,yy}=-1\quad\text{on}~(0,1),~\qquad U_{0}(0)=U_{0}(1)=0,
\end{equation*}
with $D_{0}=D(x_{k})$. This yields%
\begin{equation*}
U_{0}(y)=\frac{1}{2D_{0}}(y-y^{2}).
\end{equation*}
Since $\omega$ varies on the $O(1)$ scale while $l_{k}=O(1/N)$, it is
constant across an interval at leading order, so that in original variables
the contribution of $[x_{k},x_{k}+l_{k}]$ to $\overline{u}_{\omega}$ in (\ref%
{ubarw1D}) is 
\begin{equation*}
\int_{x_{k}}^{x_{k}+l_{k}}\omega
u\,dx\sim\omega(x_{k})\,l_{k}^{3}\int_{0}^{1}U_{0}(y)\,dy=\frac{%
\omega(x_{k})\,l_{k}^{3}}{12D(x_{k})}\,.
\end{equation*}
It follows that optimal trap locations correspond to the minimum of the
energy 
\begin{equation}
E=\sum\limits_{k=0}^{N}\omega(x_{k})\frac{l_{k}^{3}}{D(x_{k})}\,,\qquad
l_{k}=x_{k+1}-x_{k}.  \label{19:21}
\end{equation}
From Proposition \ref{prop:cont} with $f(\ell,x)=\omega\ell^{3}/D$ we obtain 
$\frac{\omega\ell^{3}}{D}-\frac{3\omega\ell^{3}}{D}=-\frac{2\omega\ell^{3}} {%
D}=\lambda,$ so that 
\begin{equation}
\ell(x)=C\left( \frac{D(x)}{\omega(x)}\right) ^{1/3},\qquad\text{with}~C=%
\frac{1}{N}\int_{0}^{L}\left( \frac{\omega(x)}{D(x)}\right) ^{1/3}dx
\end{equation}
being determined via (\ref{1855}). 
The effective trap density is given by $\rho(x)\sim1/\ell(x)$, so that 
\begin{equation}
\rho(x)\propto\left( \omega(x)M(x)\right) ^{1/3}=\left( \frac{\omega (x)}{%
D(x)}\right) ^{1/3}.  \label{rho1D}
\end{equation}
In particular, a uniform starting distribution $\omega\equiv1/L$ gives $%
\rho\propto M^{1/3}=D^{-1/3}$, while the stationary choice $\omega
=M/\left\vert M\right\vert $ gives $\rho\propto M^{2/3}$ -- in contrast to
2D and 3D, where that same choice yielded $\rho\propto M$ for every $\nu N$.

We remark that this result is independent of the choice of $\theta$ (It\^{o}
vs. Stratonovich).

\textbf{Example. }Suppose that 
\begin{equation}
D(x)=\exp(-ax),
\end{equation}
and take the starting distribution to be uniform, $\omega\equiv1/L$. Then 
\begin{equation}
\ell(x)=\frac{3}{Na}\left( \exp\left( La/3\right) -1\right) \exp(-ax/3).
\label{19:37}
\end{equation}

Figure \ref{fig:D} shows a comparison between a full numerical optimization
of the problem (\ref{19:21}) with $a=1$ and $L=4$, versus a full
optimization of (\ref{1733}), and the asymptotic result (\ref{19:37}). Note
that the distribution of the traps is denser on the right, where the
diffusion coefficient $D(x)$ is smaller. Good agreement is observed, which
improves for larger $N.$

\begin{figure}[ptb]
\includegraphics[width=0.48\textwidth]{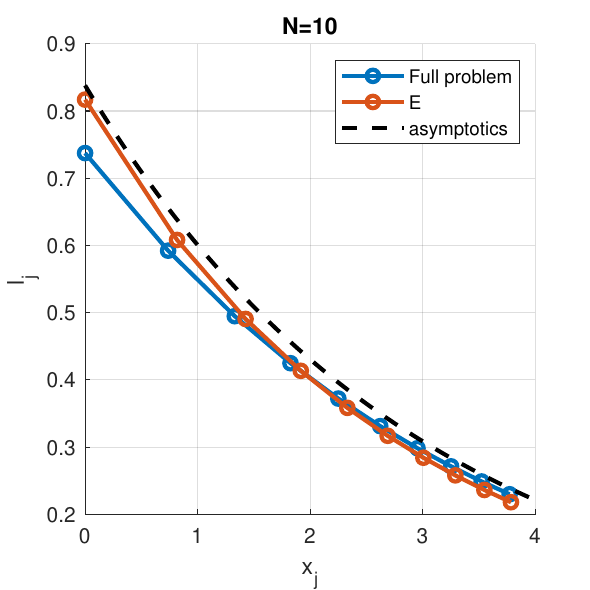} \includegraphics[width=0.48%
\textwidth]{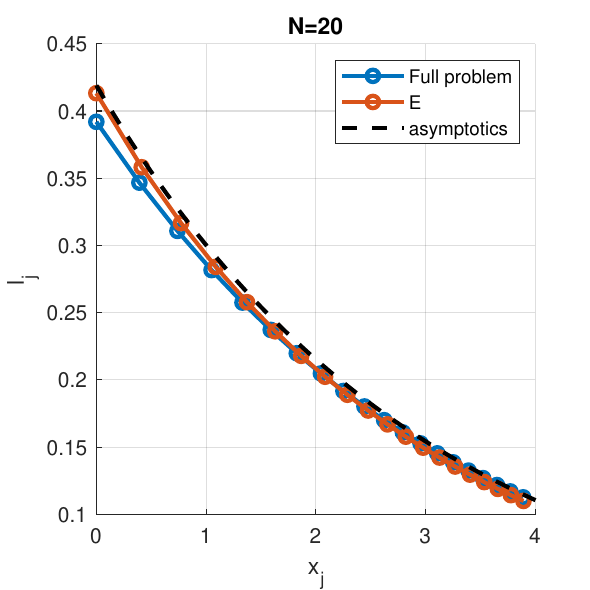}
\caption{ Optimal trap distribution which minimizes the global MFPT (uniform 
$\protect\omega$) for the problem (\protect\ref{1733}) with $D(x)=\exp(-x)$, 
$\protect\theta=0$, $L=4$, and with $N=10$ traps (left) and $N=20$ traps
(right). The graphs show a comparison between the numerical minimization of
the original problem (\protect\ref{1733}) (labeled ``full problem''),
numerical minimization of the simplified energy $E$ given by (\protect\ref%
{19:21}) (labeled ``E''), and the asymptotic formula (\protect\ref{19:37})
(labeled ``asymptotics''). }
\label{fig:D}
\end{figure}

\subsection{Heterogeneous drift or thin domain}

Consider a particle undergoing a Brownian motion in a one-dimensional medium
with space-dependent velocity $v(x)$. The associated Fokker--Planck equation
reads 
\begin{equation}
p_{t}+\left( vp\right) _{x}= D_{0} p_{xx}.  \label{fp}
\end{equation}
Our goal is to distribute $N$ traps at locations $x_{j}$ within a domain $%
[0,L]$ to minimize the mean capture time. Rescaling $D_{0}=1 $, the
associated MFPT\ problem becomes 
\begin{equation}
u_{xx}+v(x)u_{x}=-1 \quad\text{on}~ (x_{j},x_{j+1}), \qquad u(x_{j})=0 ,
\qquad j=0\ldots N.  \label{1883}
\end{equation}

Note that the same problem appears in the context of thin domains of
variable thickness, with traps that transversally span the domain. A
standard approximation via the integral over the transversal direction leads
to the problem (\ref{1883}) with $v(x)=\frac{\sigma^{\prime}(x)}{\sigma(x)}$%
, where $\sigma(x)$ is the cross-sectional area along the transversal
direction (see \cite{Grebenkov2020,Grebenkov2022} and references therein).

As in \S \,\ref{sec:D}, here we are interested in finding optimal trap
locations minimizing the $\omega$-averaged MFPT $\overline{u}_{\omega}$ from
(\ref{ubarw1D}). We consider two regimes: $v(x)=O(1)$ and $v(x)=O(N)$.

\textbf{Case A:}$\ v(x)=O(1)$. We proceed as in \S \,\ref{sec:D}. Expand
along an interval $[x_{k},x_{k}+l_{k}]$, with $l_{k}\ll1$. Rescaling $%
u(x)=l_{k}^{2}U(y)$, $x=x_{k}+yl_{k},$ we obtain, at leading order, that $%
U_{yy}\sim-1$. As such, the minimization problem is the same as (\ref{19:21}%
) with $D(x)=1$, so that (\ref{rho1D}) gives the leading-order density 
\begin{equation}
\rho(x)\propto\omega^{1/3}(x),\qquad\ell(x)\propto\omega^{-1/3}(x).
\label{rhoAlead}
\end{equation}
The drift therefore has no effect at the leading order. Note that this is
typically the case for a thin domain, where the geometry is independent of $%
N.$

\textbf{Case B:}\ $v(x)=O(N).$ Assume as before that $x_{k+1}-x_{k}=O(1/N).$
Let $x_{k+1}=x_{k}+l_{k}$, where $l_{k}=O(1/N)$. Along an interval $%
[x_{k},x_{k}+l_{k}]$, with $l_{k}\ll1$, we rescale $u(x)=l_{k}^{2}U(y)$, $%
x=x_{k}+yl_{k}$ to obtain the problem 
\begin{equation}
U_{yy}+v_{0}U_{y}\sim-1\quad\text{on}~(0,1),\qquad U(0)=0=U(1),
\end{equation}
where $v_{0}=l_{k} v(x_{k})=O(1)$. After some algebra we then obtain 
\begin{equation*}
F(v_{0}):=\int_{0}^{1}U(y)dy=-\frac{1}{v_{0}^{2}}+\frac{1}{2v_{0}}\coth
(v_{0}/2).
\end{equation*}
Consequently, weighting each interval by $\omega(x_{k})$ as in (\ref{19:21}%
), we seek to minimize the energy 
\begin{equation}
E=\sum\limits_{k=0}^{N}\omega(x_{k})\,l_{k}^{3}F(v(x_{k})l_{k}).  \label{EB}
\end{equation}
Equation (\ref{1854}) of Proposition \ref{prop:cont} is equivalent to 
\begin{equation*}
\omega\frac{z^{3}}{v^{3}}\left( 2F(z)+zF^{\prime}(z)\right) =\lambda,\qquad
z=\ell v.
\end{equation*}
After some algebra we find 
\begin{equation*}
2F(z)+zF^{\prime}(z)=\frac{1}{2z}\frac{\sinh z-z}{\cosh z-1} \,,
\end{equation*}
so that 
\begin{equation*}
\frac{\sinh z-z}{\cosh z-1}z^{2}=C\frac{v^{3}}{\omega},\qquad z=v\ell.
\end{equation*}
This yields a parametric solution%
\begin{equation}
\ell=C\left( \frac{z}{\omega}\,\frac{\cosh z-1}{\sinh z-z}\right)
^{1/3},\qquad v=z/\ell,  \label{lB-general}
\end{equation}
where the constant $C$ is chosen according to (\ref{1855}). 
As a check, $\frac{\cosh z-1}{\sinh z-z}\sim3/z$ as $z\rightarrow0$, so that 
$\ell\rightarrow C\left( 3/\omega\right) ^{1/3}$, recovering the $%
\rho\propto\omega^{1/3}$ of (\ref{rhoAlead}).

For example, take $\omega\equiv1/L$ and set $z=x$, $L=6$; the constant $%
\omega$ is absorbed into $C$, leaving 
\begin{equation}
\ell=\frac{C_{0}}{N}x^{1/3}\left( \frac{\cosh x-1}{\sinh x-x}\right)
^{1/3},\ \ \ v=\frac{N}{C_{0}}\left( \frac{\sinh x-x}{\cosh x-1}x^{2}\right)
^{1/3},  \label{lB}
\end{equation}
where 
\begin{equation}
C_{0}=\int_{0}^{6}x^{-1/3}\left( \frac{\cosh x-1}{\sinh x-x}\right)
^{-1/3}dx\approx3.7883.
\end{equation}

Figure \ref{fig:v} shows excellent agreement between the formula (\ref{lB})\
and full numerics, as $N$ is increased.

\begin{figure}[ptb]
\includegraphics[width=0.48\textwidth]{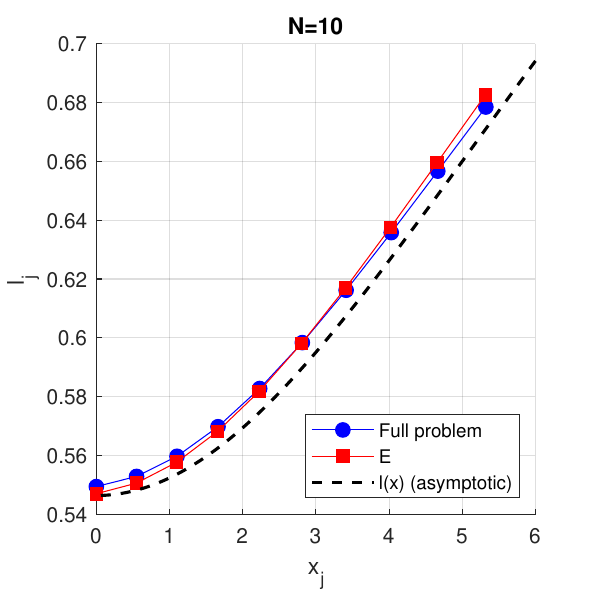} %
\includegraphics[width=0.48\textwidth]{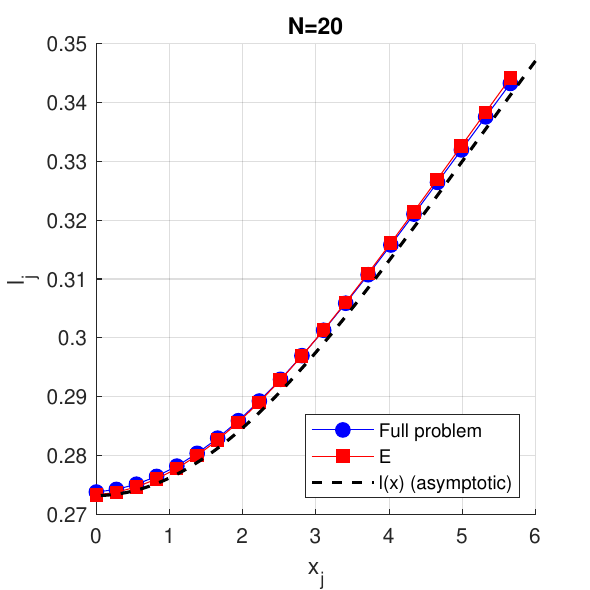}
\caption{ Optimal trap distribution which minimizes $\protect\int%
_{0}^{L}u(x)dx$ (uniform $\protect\omega$) for the problem (\protect\ref%
{1883}) with $v(x)$ as given in (\protect\ref{lB}), $L=6$, and with $N=10$
traps (left) and $N=20$ traps (right). The graphs show a comparison between
the numerical minimization of the original problem (\protect\ref{1883})
(labeled ``full problem''), numerical minimization of the simplified energy $%
E$ given by (\protect\ref{EB}) (labeled ``E''), and the asymptotic formula (%
\protect\ref{lB}) (labeled ``asymptotics''). }
\label{fig:v}
\end{figure}

\section{Discussion}

\label{sec:discussion}

In this paper we have obtained the asymptotics of the optimal trap density
that minimizes the MFPT $\overline{u}_{\omega}$ for a generic heterogeneous
diffusivity $D(x)$ and starting distribution $\omega$ of Brownian particles.

The key quantity is the trapping strength $\kappa=N\nu$ defined in (\ref{nu}%
); the optimal density is the arithmetic average $\rho=\frac{1}{2}%
\left(\mu+\omega\right) $ for weak trapping $\kappa\ll1$ and is the
geometric average $\rho\propto\sqrt{\mu\omega}$ for strong trapping $%
\kappa\gg1$. When $\omega=\mu$ (corresponding to the equilibrium particle
distribution in the absence of traps), we obtain $\rho=\mu$ for \emph{any} $%
\kappa$, with $u(x)$ being constant in this case. In 2D, the
\textquotedblleft classical\textquotedblright\ approach using Green's
functions works only in a very restricted regime where $1\ll
N\ll\log\varepsilon^{-1}$; the restriction is milder in 3D\ but still
significant ($N\ll\varepsilon^{-1} $). For a more realistic regime $%
N\gtrsim\log\varepsilon^{-1}$, we use a homogenization technique and derive
the optimality system (\ref{kappa-any}) that determines the optimal density;
the geometric average occurs precisely in the asymptotic limit $%
N\gg\log\varepsilon^{-1}.$ Numerical simulations of the full problem (\ref%
{u2d})\ are shown to agree with the homogenization theory.

The results also apply to homogeneous diffusion, but with a non-homogeneous
initial distribution. For example, suppose that the particles start
uniformly in the left half of the domain, and diffusion is homogeneous. In
the weak-trapping limit $\kappa\ll1$, the optimal trap distribution from (%
\ref{rhoavg}) corresponds to placing 75\% of the traps on the left, and 25\%
on the right, both uniformly. The strong-trapping limit $\kappa\gg1$ (\ref{rhogeom}), by contrast,  calls for placing all of the traps on the left
side of the domain and none on the right. Figure \ref{fig:kappa} (top)\ is
an illustration of that, but with the roles of $\omega$ and $\mu$ reversed
(note that all the density formulas are symmetric with respect to swapping $%
\mu\leftrightarrow\omega$). Interestingly, these results do not ``feel'' the
shape of the domain at leading order, at least in the strong- or
weak-trapping limits.

Let us contrast the optimal trap density with the uniform density. Define
the $\func{gain}=\overline{u}_{\omega }[\rho _{\text{uni}}]/\overline{u}%
_{\omega }[\rho _{\text{opt}}]$, by which an optimal trap arrangement is
better than the uniform one; here we use $\rho _{\mathrm{opt}}\ $to denote
the optimal density $\rho $ as found by solving the optimality system (\ref%
{kappa-any}). Table~\ref{tab:comparison} gives $\func{gain}$ as a function
of $\kappa $ for three examples: the equilibrium start $\omega =\mu $ with
the space-dependent diffusivity of Figures \ref{fig:kappa} and \ref%
{fig:sphere3d}, and a homogeneous medium on $\Omega =(0,1)$ with the
particles released uniformly in the left half.

\begin{table}[ptb]
\begin{center}
\begin{tabular}{lccccccccccc}
\hline
$\kappa$ & $0.1$ & $0.5$ & $1$ & $5$ & $10$ & $50$ & $100$ & $500$ & $1000$
& $5000$ & $\infty$ \\ \hline
$\func{gain}$ for $\omega=\mu$, rectangle of Fig. \ref{fig:kappa} & $1.022 $
& $1.098$ & $1.172$ & $1.449$ & $1.563$ & $1.715$ & $1.743$ & $1.770$ & $%
1.774$ & $1.777$ & $1.778$ \\ 
$\func{gain}$ for $\omega=\mu$, ball of Fig. \ref{fig:sphere3d}  & $1.003$ & $1.013$ & $1.026$ & $1.127$ & $1.241$ & $1.870$ & $2.301$
& $3.224$ & $3.468$ & $3.731$ & $3.814$ \\ \hline
$\func{gain}$ for $\Omega=(0,1)$, $\mu=1$, $\omega=2\cdot \mathbf{1}_{x<1/2}$
& $1.002$ & $1.010$ & $1.020$ & $1.082$ & $1.139$ & {$1.344$} & $1.449$ & $%
1.677$ & $1.755$ & {$1.879$} & $2$ \\ \hline
\end{tabular}%
\end{center}
\caption{Gain of optimal trap configuration vs. uniform, for different
settings.}
\label{tab:comparison}
\end{table}

When starting from the equilibrium particle density, $\omega =\mu ,$ we have 
$\overline{u}_{\omega }[\kappa ;\rho _{\text{opt}}]=\frac{\left\vert
M\right\vert }{\kappa }$ by (\ref{mainresult}), whereas $\overline{u}%
_{\omega }[\kappa ;\rho _{\text{uni}}]$ is obtained by solving (\ref{homog})
with $\rho \equiv 1/\left\vert \Omega \right\vert $ and evaluating $%
\int_{\Omega }\mu u.$ Moreover $\func{gain}\rightarrow 1$ as $\kappa
\rightarrow 0$, since then $u\sim \left\vert M\right\vert /\kappa $
irrespective of $\rho $; {\ in turn, as $\kappa \rightarrow \infty $, we
have $u\sim M/(\kappa \rho )$ and thus} 
\begin{equation}
\func{gain}(\infty )=\frac{\left\vert \Omega \right\vert \int_{\Omega }M^{2}%
}{\left( \int_{\Omega }M\right) ^{2}}\geq 1.  \label{dinf}
\end{equation}%
By Cauchy--Schwarz, the equality can only hold for constant $M$ and the gain
is highest for highly concentrated $M$. For the examples of Figures \ref%
{fig:kappa} and \ref{fig:sphere3d}, we find that the optimal placement is $%
71\%$ better than a uniform arrangement of the same traps for the rectangle (%
$\kappa =42$, $\func{gain}=1.71$) and $44\%$ better for the ball ($\kappa =20$, $\func{gain}=1.44$). The gain increases faster with $%
\kappa $ for the rectangle at first but ultimately is more pronounced for
the ball example for even larger $\kappa $.

The third row is the homogeneous example discussed above, with the particles
released in the left half only. Here $\rho_{\text{uni}}$ gives $u\equiv
1/\kappa$ exactly, and $\rho_{\text{opt}}$ is computed by direct
minimization of $\overline{u}_{\omega}$ over $\rho$, see Appendix \ref%
{app:numerics}. The last entry is the strong-trapping limit (\ref{rhogeom}),
which puts every trap in the left half and gives $\func{gain}(\infty)=1/(\int%
\sqrt {\omega\mu})^{2}=2$. This limit is approached very slowly, likely due
to the discontinuity of $\omega$ at $x=1/2$. Numerics for high $\kappa$
indicate a \textquotedblleft rampart\textquotedblright\ of high density near
the interface $x=1/2$, suggesting an optimal placement of traps akin to
\textquotedblleft border guards\textquotedblright\ on the boundary of the
area where the particles start. We leave further investigations into this
phenomenon as an open problem.

An interesting question is how a transition from the 1D scaling ($\rho
\propto\left( \omega\mu\right) ^{1/3}$) to a 2D scaling ($\rho \propto\left(
\omega\mu\right) ^{1/2}$ or $\rho\propto\omega+\mu$) occurs. Consider a thin
strip of width $O(h)$, with $h \ll1$. If we take traps that span the entire
transversal direction, then their distribution satisfies a 1D scaling $%
\rho\propto\left( \omega\mu\right) ^{1/3}$. On the other hand, if we take
traps that are disks whose radius $\varepsilon \ll h$, we obtain a 2D
scaling, as shown for example in Figure \ref{fig:rect2} (top row, $N=5$).
What if we take disks of radius commensurate with the channel width, $%
\varepsilon=O(h)$? One can expect a transition between different scaling
regimes as the disk radius is increased. This is an interesting avenue for
future exploration.

We showed that the 1D scaling $\rho \propto \left( \omega \mu \right) ^{1/3}$
is independent (to leading order) of the choice of the stochastic
interpretation (It\^{o} vs. Stratonovich), i.e., of the parameter $\theta $
in (\ref{lawley}). At the same time, our 2D analysis explicitly requires $%
\theta =0$. An interesting open problem is to derive the 2D scaling laws
when $\theta >0$. {\ In this case, the equilibrium particle distribution
without traps is $p_{\mathrm{eq}}\propto \lbrack D(x)]^{\theta -1}\propto
\mu ^{1-\theta }$, which is different from $\mu $ for $\theta >0$. In
particular, the initial particle distribution $\omega =\mu $ does not
necessarily lead to the optimal trap density.} Finally, it would be
interesting to extend the 1D results for particles with drift to a
two-dimensional setting.

In conclusion, heterogeneity does not simply attract traps toward regions of
low diffusivity. The optimal distribution reflects a competition between
where particles accumulate and where they are initially supplied. This
competition leads to arithmetic averaging in the weak-trapping regime and
geometric averaging in the strong-trapping regime. The equilibrium
preparation $\omega=\mu$ is exceptional: these two principles coincide,
yielding an optimal trap distribution identical to the particle distribution
and a spatially uniform MFPT for arbitrary trapping strength.

\subsection*{Acknowledgements}

We are grateful to Sean Lawley for interesting discussions and for
suggesting to look at a non-uniform $\omega $. {Claude AI\ (Opus version 5)
significantly assisted the elaboration of homogenization technique in \S \ref%
{sec:homog}, the optimality system for general $\kappa $ (\ref{kappa-any}),
as well as Theorems \ref{thm:minimax}--\ref{thm:general}. It was also used
for developing numerics for the full problem (\ref{u2d}) in 2D to validate
the \textquotedblleft hybrid\textquotedblright\ system and density
optimization (Fig. \ref{fig:kappa}), as well as for creating a web
application showcasing 2D results visually \cite{webapp}. 
The algorithms and the reported numerical results for the full problem in 2D
have been independently verified by the authors. Appendix A\ was written by
AI and edited by authors, whereas the remaining text} was written by
authors. 
The authors assume full responsibility for all content.


\appendix
\renewcommand\thesubsection{\thesection.\arabic{subsection}}

\section{Details of the numerical implementation}

\label{app:numerics}


\emph{The full problem in 2D.} Equation (\ref{u2d}) is discretized by finite
differences on a cell-centred Cartesian grid. The reflecting condition on $%
\partial \Omega $ is imposed through the mirror ghost value $u_{\text{ghost}%
}=u_{P}$, which on a cell-centred grid is exactly no-flux and second-order
accurate. The trap boundaries are resolved by the Shortley--Weller
construction \cite{shortley1938numerical}: on a grid line cut by $\partial
B_{\varepsilon }(x_{j})$ the regular spacing $h$ is replaced by the true
distance to the interface and the Dirichlet value moved to the right-hand
side, which avoids staircasing the traps; cut fractions are clamped from
below to keep the matrix well conditioned. The resulting non-symmetric
sparse matrix is factorized directly. The integral $\int_{\Omega }\omega u$
was computed by using the midpoint rule weighted by the fraction of each
cell's area that lies outside the traps; without this weighting the
quadrature loses an order of accuracy along every trap boundary.

\emph{Optimization over trap locations in 2D.} The objective $J(x_{1},\ldots
,x_{N})=\int_{\Omega }\omega u$ is minimized by gradient descent, each trial
step being followed by a projection of the centers back to at least $%
1.2\varepsilon $ from every wall and $2.6\varepsilon $ from one another,
center to center. The gradient is obtained from the adjoint. Introduce $v$,
the solution of the same boundary value problem as $u$ but with the source
density $\omega $ in place of $M$: 
\begin{equation}
\Delta v=-\omega \ \ \text{in }\Omega \backslash \left( \cup
_{j}B_{\varepsilon }(x_{j})\right) ,\qquad v=0\ \ \text{on }\cup
_{j}\partial B_{\varepsilon }(x_{j}),\qquad \partial _{n}v=0\ \ \text{on }%
\partial \Omega .  \label{adjoint2d}
\end{equation}%
This is the discrete counterpart of the adjoint $v$ of the optimality system
(\ref{kappa-any}): there $v$ carries the sensitivity of $\overline{u}%
_{\omega }$ to the homogenized density $\rho $, here to the position of an
individual trap. Moving $\partial B_{\varepsilon }(x_{k})$ with normal
velocity $V_{n}$ changes $J$ by the Hadamard shape derivative for a
Dirichlet boundary \cite{sokolowski1992introduction}, 
\begin{equation}
\delta J=-\int_{\partial B_{\varepsilon }(x_{k})}\partial _{n}u\,\partial
_{n}v\,V_{n}\,ds,  \label{hadamard}
\end{equation}%
where $n$ is the outward normal of the domain, i.e.\ pointing into the trap.
Translating the trap rigidly by a unit vector $e$ means $V_{n}=e\cdot n$,
and on the circle $ds=\varepsilon \,d\vartheta $, so that taking $e$ along
each coordinate direction in turn gives the two components of 
\begin{equation}
\frac{\partial J}{\partial x_{k}}=-\varepsilon \int_{0}^{2\pi }\partial
_{r}u\,\partial _{r}v\,(\cos \vartheta ,\sin \vartheta )\,d\vartheta ,
\label{shapederiv2d}
\end{equation}%
the integral being over $\partial B_{\varepsilon }(x_{k})$. Equation (\ref%
{shapederiv2d}) is the exact, finite-$\varepsilon $ counterpart of the
gradient of the Green's energy (\ref{greenE}): expanding it as $\varepsilon
\rightarrow 0$ reproduces $\partial E/\partial x_{k}$, and hence the
velocity field of the flow (\ref{ode}).

The step length is chosen by the two-point rule of Barzilai and Borwein \cite%
{barzilai1988two}. Writing $s^{(i)}=x^{(i)}-x^{(i-1)}$ and $%
y^{(i)}=g^{(i)}-g^{(i-1)}$ for the changes in the configuration and in its
gradient between successive iterates, 
\begin{equation}
t^{(i)}=\frac{s^{(i)}\cdot s^{(i)}}{s^{(i)}\cdot y^{(i)}},\qquad
x^{(i+1)}=x^{(i)}-t^{(i)}g^{(i)},  \label{bb}
\end{equation}%
whenever $s^{(i)}\cdot y^{(i)}>0$, and $t^{(i)}=2t^{(i-1)}$ otherwise; the
first step is scaled so that the largest trap moves a prescribed small
distance. The quotient in (\ref{bb}) is a Rayleigh quotient of the Hessian
along the last step, so it carries curvature information without any second
derivative ever being formed. The raw Barzilai-Borwein iteration is not
monotone, so each step is accepted only if it satisfies an Armijo condition 
\cite{armijo1966minimization}, $J(x-tg)\leq J(x)-c_{1}t\left\Vert
g\right\Vert ^{2}$ with $c_{1}=10^{-4}$, and the step is halved until it
does; this costs one extra objective evaluation per backtrack, and makes the
objective history decrease monotonically. A nonmonotone acceptance window
was tried and was slower on these problems.

\emph{The ball.} All three-dimensional results are computed from the hybrid
system (\ref{140}): for a given configuration, $T$ and $C_{1},\ldots ,C_{N}$
follow from a single $(N+1)\times (N+1)$ solve, which delivers $\overline{u}%
_{\omega }$ through (\ref{fullubar}). The trap centers are then optimized by
the same Barzilai--Borwein descent (\ref{bb}) with Armijo backtracking, the
gradient being computed analytically from the adjoint of that linear system.
No constraint keeping the traps apart or inside the ball is needed here: the
objective is its own barrier, since $R(x,x)$ diverges at the boundary and $%
G(x_{j},x_{k})$ as two traps merge. We minimize the Green's energy (\ref%
{greenE}) first, from random initial conditions, and use its minimizer as
the starting point for the hybrid objective; the descent is stopped when $%
\max_{k}\left\vert \partial \overline{u}_{\omega }/\partial x_{k}\right\vert 
$ falls below $10^{-8}$. Effective densities are estimated from Voronoi cell
volumes clipped to the unit ball. As an independent check, (\ref{u2d}) was
also solved directly in the ball by a cell-centred finite-volume scheme on a
uniform Cartesian grid; resolving the traps needs $n=8/\varepsilon $ cells
per side, so this is affordable only at a relatively large $\varepsilon$:
around 8 million cells were needed for $\varepsilon=0.04$ which we used for
validation; with $N=20$ the hybrid and direct solve agreed to within $1\%$.

\emph{The optimal density at finite} $\kappa $. The solid curves of Figure %
\ref{fig:sphere3d} and the third row of Table \ref{tab:comparison} are
minimizers of $\overline{u}_{\omega }=\int_{\Omega }\omega u$ over densities 
$\rho \geq 0$ with $\int_{\Omega }\rho =1$, where $u$ solves the homogenized
problem (\ref{homog}). In all considered cases, $M$ and $\omega $ depend on
one coordinate only, so (\ref{homog}) is a two-point boundary value problem,
either in the radial coordinate $r$ for the ball, or in $x$ for the
rectangle; we discretize it by cell-centred finite volumes ($n=400$ cells,
zero flux at both ends), so that one evaluation of $\overline{u}_{\omega }$
is one tridiagonal solve.

The gradient is available from the adjoint: with $v$ solving (\ref{homog})
with $\omega $ in place of $M$, a variation $\delta \rho $ changes the
objective by $\delta \overline{u}_{\omega }=-\kappa \int_{\Omega }uv\,\delta
\rho $, so the descent direction is $\kappa uv$, at the cost of a second
tridiagonal solve. This is the same quantity that appears in the optimality
system (\ref{kappa-any}), whose condition $uv=\mathrm{const}$ on the support
of $\rho $ is exactly stationarity of $\overline{u}_{\omega }$ under the
mass constraint. Positivity and normalization are built in by optimizing
over $\theta (x)$ with $\rho (x)=e^{\theta (x)}/\int_{\Omega }e^{\theta }$,
which turns the constrained problem into an unconstrained one; the chain
rule gives $\partial \overline{u}_{\omega }/\partial \theta _{i}=\rho _{i}%
\bigl(G_{i}-\textstyle\sum_{j}w_{j}\rho _{j}G_{j}\bigr)$ with $G_{i}=-\kappa
w_{i}u_{i}v_{i}$ and $w_{i}$ the cell volumes, i.e.\ the gradient with its $%
\rho $-weighted mean removed.

The descent itself is the Barzilai--Borwein iteration described above, (\ref%
{bb}) with Armijo backtracking, applied to $\theta$: no Hessian is formed,
and the two-point step carries the curvature information that a plain
gradient method would lack. We continue in $\kappa$, starting from the
closed-form small-$\kappa$ optimum (\ref{rhoavg}) and using the minimizer at
each $\kappa$ as the starting point for the next; the iteration is stopped
when the gradient norm falls below $10^{-6}$ relative to the objective,
which takes between a few and two hundred steps per $\kappa$, and an entire $%
\kappa$ column of Table \ref{tab:comparison} takes well under a minute. For $%
\omega=\mu$ the method returns $\rho=\mu$ and $\overline{u}%
_{\omega}=\left\vert M\right\vert /\kappa$ to the tolerance used, as (\ref%
{mainresult}) requires, and doubling $n$ changes the reported values of the
gain in the fourth digit at most.

\bibliographystyle{elsarticle-num}
\bibliography{bib}

\end{document}